\documentclass[11pt]{article}

\usepackage{amsmath}
\usepackage{amssymb}
\usepackage{amsthm}
\usepackage{mathtools}
\usepackage{graphicx}
\usepackage{caption}

\usepackage{natbib}

\usepackage{xcolor}
\allowdisplaybreaks

\newtheorem{theorem}{Theorem}[section]

\newtheorem{corollary}[theorem]{Corollary}

\newtheorem{lemma}[theorem]{Lemma}

\newtheorem{proposition}[theorem]{Proposition}
\newtheorem{remark}[theorem]{Remark}

\usepackage{anysize}
\marginsize{20mm}{20mm}{15mm}{30mm} % left right top bottom
\newcommand{\myfootnote}[1]{\linespread{1.0}\selectfont\footnote{#1}\linespread{1.2}\selectfont}

\DeclareMathOperator{\Var}{Var}
\DeclareMathOperator{\Cov}{Cov}

\DeclareMathOperator{\tr}{tr}
\newcommand{\smallpar}{\par\vspace{\baselineskip}}
\begin{document}

%\begin{frontmatter}
\title{
Joint survival annuity derivative valuation in the linear-rational Wishart mortality model
}
\author{Jos\'{e} Da Fonseca\thanks{Auckland University of Technology, Business School, Department of Finance, Private Bag 92006, 1142 Auckland, New Zealand. Phone: +64 9 9219999 extn 5063. Email: jose.dafonseca@aut.ac.nz and PRISM Sorbonne EA 4101, Universit\'e Paris 1 Panth\'eon - Sorbonne, 17 rue de la Sorbonne, 75005 Paris, France. ORCID: 0000-0002-6882-4511} 
\and Patrick Wong\thanks{Monash University, Business School, Department of Econometrics and Business Statistics, Melbourne, VIC, 3800, Australia. Email: patrick.wong@monash.edu. ORCID: 0000-0002-6164-901X}
}

\date{\today}
%\fontsize{12pt}{15pt}\selectfont

\maketitle

\begin{abstract}
This study proposes a linear-rational joint survival mortality model based on the Wishart process. The Wishart process, which is a stochastic continuous matrix affine process, allows for a general dependency between the mortality intensities that are constructed to be positive. Using the linear-rational framework along with the Wishart process as state variable, we derive a closed-form expression for the joint survival annuity, as well as the guaranteed joint survival annuity option. Exploiting our parameterisation of the Wishart process, we explicit the distribution of the mortality intensities and their dependency. We provide the distribution (density and cumulative distribution) of the joint survival annuity. We also develop some polynomial expansions for the underlying state variable that lead to fast and accurate approximations for the guaranteed joint survival annuity option. These polynomial expansions also significantly simplify the implementation of the model. Overall, the linear-rational Wishart mortality model provides a flexible and unified framework for modelling and managing joint mortality risk. 
\end{abstract}

\vspace{6cm}
\begin{bf}JEL Classification\end{bf}: G12; G13; C61

\textbf{Acknowledgements:} We thank Qihe Tang and Francesco Ungolo from UNSW Sydney for their useful suggestions and remarks. We also thank the participants at IME 2023, Edinburgh, United Kingdom, and Longevity 19, Amsterdam, Netherlands for their remarks. The usual caveat applies.

%\vspace{1cm}
%\linespread{1}
%\selectfont

\clearpage
\linespread{1.2}
\selectfont
\section{Introduction}

In an era of increasing life expectancies, longevity products such as life annuities have become vital income sources for individuals in their retirement years. The complexities involved in modeling these products cannot be overstated, especially when dealing with multi-life contracts that introduce intricate dependencies among the insured parties. This paper presents a novel framework for modeling dependent mortality intensities and pricing longevity products, including joint survival annuities and options on these annuities. Taking into account our objectives, we use the stochastic differential equations (SDEs) framework  outlined by \citet{Dahl2004}. This method, further explored by \citet{Biffis2005}, is built on the vector affine process originally introduced by \citet{DuffieKan1996}. This framework, initially introduced for modelling a single mortality intensity, has been extended to the case of multiple annuitants/cohorts/populations in several works, see for example \citet{JevticLucianoVigna2013}, \citet{JevticHurd2017}, \citet{JevticRegis2019},  \citet{ZeddoukDevolder2020}, \citet{XuSherrisZiveyi2020b}, \citet{HuangSherrisVillegasZiveyi2022}, \citet{UngoloGarcesSherrisZhou2023}, \citet{Hainaut2023} or \citet{XuChenYang2024} to name a few. However, when it comes to pricing options on multi-annuitant products, such as a joint survival annuity, there is, to the best of our knowledge, no work available. This is largely due to an inherent difficulty in pricing from the exponential vector affine framework used in the aforementioned works.
To overcome this difficulty, we depart from the exponential vector affine literature by first using the Wishart process, an SDE with values in the space of positive definite matrices, a natural space for capturing dependencies. We also depart from this literature by using the potential approach of \citet{Rogers1997} to build a multi-annuitant mortality model for which the joint survival bond is a linear-rational function of the Wishart process. Thanks to this linear-rational structure, the density of the joint survival annuity can be computed using a one-dimensional integration and thus also provides us with key risk management indicators such as value at risk and expected shortfall. It also allows the exact pricing of the guaranteed joint survival annuity option (GAO), as well as fast and efficient GAO price approximations. We provide a complete model implementation that illustrates the analytical results and shows the impact of mortality dependence on the GAO.\\

In the framework introduced in \citet{Dahl2004}, the (joint) survival bond is an exponential affine function of the process and, as a consequence, the future annuity value, which is a sum of (joint) survival bonds, does not admit a known density. This presents two main problems. First, it makes exact pricing of GAOs impossible. The absence of an exact option pricing formula in mortality products and interest-rate products is well known and explained in detail in \citet{BiffisMillossovich2006}. The second problem of the conventional approach stems from a property inherent to the vector affine process of \citet{DuffieKan1996}. Namely, the diffusion of this process must adhere to specific constraints, as detailed in \citet{DuffieFilipovicSchachermayer2003}. When applied to model dependencies, this can result in certain correlation limitations, as illustrated by \citet{DaFonsecaGrasselliTebaldi2007} in the context of an equity derivative. These are the two main problems with the conventional approach when it comes to pricing options on a (joint) survival annuity and modelling the dependence between mortality intensities, and our model will address both of them.\\

Following the potential approach of \citet{Rogers1997}, we build a multi-annuitant mortality model for which the joint survival bond is a linear-rational function in the state variable. For the state variable, we use the Wishart process, which is a stochastic process with values in the space of positive definite matrices (\textit{i.e.}, covariances) that has strong analytical properties. \citet{GrasselliTebaldi2008} show how to compute its moment generating function in closed form, it depends on a matrix Riccati differential equation that can be explicitly integrated (it is solvable), making the Wishart process as tractable as the Cox-Ingersoll-Ross process (\textit{i.e.,} the scalar square root process).  The Wishart process is used in \citet{DaFonsecaGrasselliTebaldi2007} to build a multi-asset equity model with dependent equity volatilities; in \citet{ChiarellaDaFonsecaGrasselli2014} to build a multi-factor interest rate model with correlated factors to price range notes derivatives; in \citet{DeelstraGrasselliVanWeverberg2016} to model the dependence between interest rates and a (single) mortality intensity to price a GAO in the exponential affine framework; and in \citet{DaFonseca2024} for modelling the dependence between interest rates and a (single) mortality intensity to price a GAO in a linear-rational framework. In all these works, the remarkable analytical properties of the Wishart process are crucial for obtaining the results. The Wishart process allows us to specify a very general dependence form between the different mortality intensities and the linear-rational structure of the model ensures that the joint survival annuity is a linear-rational function of the Wishart process. This property, combined with the analytical properties of the Wishart process, allows the calculation (up to numerical integration) of the density of the joint survival annuity and therefore the development of various risk management tools, such as value at risk or expected shortfall, which depend on this density. Furthermore, these properties also allow the calculation of the GAO in closed form, making the linear-rational Wishart model much more tractable than models based on the exponential affine modelling  approach. Even if GAO pricing can be done efficiently, the properties of the Wishart process easily yield some interesting quantities, such as the moments of the joint survival annuity, which can be used to derive GAO price approximations in the spirit of \citet{CollinDufresneGoldstein2002} or \citet{FilipovicMayerhoferSchneider2013}. We think these quantities could also be used to estimate the model using a method of moments. Beyond these analytical results, and to gain a better understanding of the model, we illustrate how the model works through several numerical experiments. One of them confirms how important the correlation between mortality intensities is when pricing a GAO.\\

The remainder of the paper is organised as follows. In Section~\ref{Framework} we introduce the instruments we intend to value, which includes the joint survival annuity and an option on the joint survival annuity, the guaranteed joint survival annuity option. Section~\ref{LRWishart} introduces the linear-rational Wishart mortality model. Here, we develop the necessary theory and background to work with our model and derive different properties of the model, in particular its characteristic function owing to the affine nature of the Wishart process. We also introduce the potential approach here, which is then used in Section~\ref{PricingAnnuity} where we price the joint survival annuity and explicitly derive the positive mortality intensities. Section~\ref{PricingAnnuityDerivative} prices the guaranteed joint survival annuity option, where we also develop three approximation approaches based on a Gaussian approximation, a spectral decomposition approach and a gamma approximation. An implementation and numerical experiments of our model are conducted in Section~\ref{Section:num_experiments}. %Section~\ref{Section:Extensions} contains some remarks on how to extend the model. 
Finally, we conclude and provide suggestions on further directions of research in Section~\ref{Conclusion}. {\color{black}The proofs of all the propositions in this work are presented in the Supplementary Appendix.}

\section{The modelling framework}\label{Framework}

Let \((\Omega, \mathcal{G}, (\mathcal{G}_t )_t, \mathbb{P})\) denote a filtered probability space where \(\mathbb{P}\) is the historical probability measure. We define \(\tau_x\) and \(\tau_y\) as future lifetimes for annuitants aged \(x\) and \(y\) respectively. These lifetimes have corresponding stochastic intensities \((\mu_x(t))_{t\geq 0}\) and \((\mu_y(t))_{t\geq 0}\). Referring to the annuitant aged \(x\) as \(x\) and aged \(y\) as \(y\), and inspired by \citet{Dahl2004} and \citet{Biffis2005}, \(\tau_x\) marks the initial jump-time of a \(\mathcal{G}\)-counting process \((N_t^x)_{t\geq 0}\), indicating if individual \(x\) has passed away \((N_t^x \neq 0)\) by time \(t\). Similarly, \(\tau_y\) indicates the same for individual \(y\). \\

Drawing from \citet{DeelstraGrasselliVanWeverberg2016}, \(\mathcal{R}_t\) and \(\mathcal{M}_t\) are filtrations produced by the interest rate and joint mortality intensity processes respectively. Let \(\mathcal{F}_t := \mathcal{R}_t \vee \mathcal{M}_t\) be the minimal \(\sigma\)-algebra that contains $\mathcal{R}_t \cup \mathcal{M}_t$. The filtration \((\mathcal{G}_t)_{t\geq 0}\) embodies the progressive information available, capturing the evolution of both state variables and whether the annuitants \(x\) or \(y\) have died at time \(t\). Hence, we set \(\mathcal{G}_t := \mathcal{F}_t \vee \mathcal{H}_t\), where \(\mathcal{H}_t:= \sigma(\{\{\tau_x \leq s\},\, \{\tau_y \leq s\}  ,\,0 \leq s \leq t\})\) is the minimum filtration wherein \(\tau_x\) and \(\tau_y\) are established stopping times. \\
% \(\mathcal{H}_t:= \sigma(\mathbf{1}_{\lbrace \tau_x\leq s;\; 0\leq s \leq  t\rbrace}) \vee \sigma(\mathbf{1}_{\lbrace \tau_y\leq s;\; 0\leq s \leq  t\rbrace})\)

In line with \citet{DeelstraGrasselliVanWeverberg2016}, we posit that \((N_t^x,N_t^y)_{t\geq 0}\) operates as a two-dimensional Cox process, driven by an \(\mathcal{F}\) sub-filtration of \(\mathcal{G}\). This implies the following survival probability on the set \(\lbrace \tau_x>t, \tau_y>t \rbrace \):
\begin{align}
\mathbb{Q}(\tau_x>T, \tau_y>T|\mathcal{G}_t)=\mathbb{E}^{\mathbb{Q}}\left[\left. e^{-\int_t^T\mu_x(s) + \mu_y(s)ds} \right\vert \mathcal{G}_t\right], \label{EQ:JointDeath}
\end{align}
where $\mathbb{Q}$ is a risk-neutral pricing measure equivalent to $\mathbb{P}$, and $\mathbb{E}^{\mathbb{Q}}\left[\cdot|\mathcal{G}_t\right]$ denotes the conditional expectation under $\mathbb{Q}$ given $\mathcal{G}_t$. Regarding the choice of $\mathbb{Q}$, we follow \citet{DeelstraGrasselliVanWeverberg2016} and references therein and assume that it can be selected on the basis of available market data. For literature addressing the non-diversifiable aspect of mortality risk, which results in an incomplete market, refer to studies such as \citet{DahlMoller2006}, \citet{BayraktarMilevskyDavidPromislowYoung2009}, \citet{HuangTsaiYangCheng2014}, and \citet{CeciColaneriCretarola2015}.\\ 

For completeness, we restate some well known equations in the literature. We note that since $(N_t^x, N_t^y)$ is a two-dimensional Cox process, then it follows that conditional on a particular trajectory $(\mu_x(t),\mu_y(t))_{t \geq 0}$ the process $(N_t^x,N_t^y)_{t \geq 0}$ is a two-dimensional independent Poisson-inhomogeneous process with parameter $(\int_0^t\mu_x(s)ds)_{t \geq 0}$ for $(N_t^x)_{t \geq 0}$ and $(\int_0^t\mu_y(s)ds)_{t \geq 0}$ for $(N_t^y)_{t \geq 0}$. As such, for all $t_1 \geq  t  \geq 0 $  and $t_2 \geq  t  \geq 0 $ and nonnegative integers $k_1$, $k_2$, the following equality holds
\begin{align*}
I(k_1,k_2)	:=& \,\mathbb{Q}( (N_{t_1}^x - N_t^x = k_1) (N_{t_2}^y- N_t^y = k_2)| \mathcal{H}_{t} \vee \mathcal{F}_{\infty})\\
			=& \,\mathbb{Q}( (N_{t_1}^x - N_t^x = k_1) | \mathcal{H}_{t} \vee \mathcal{F}_{\infty}) \mathbb{Q}( (N_{t_2}^y- N_t^y = k_2)| \mathcal{H}_{t} \vee \mathcal{F}_{\infty})\\
			=& \,\frac{\left(\int_t^{t_1}\mu_x(s)ds\right)^{k_1} }{k_1!}e^{-\int_t^{t_1}\mu_x(s)ds}\frac{\left(\int_t^{t_2}\mu_y(s)ds\right)^{k_2} }{k_2!}e^{-\int_t^{t_2}\mu_y(s)ds},
\end{align*}
where we used the conditional independence between $(N_t^{x})_{t\geq 0}$ and $(N_t^{y})_{t\geq 0}$. Taking $k_1=0$, $k_2=0$ and integrating the intensities lead to \eqref{EQ:JointDeath}.\\

Still following the literature, see for example \citet[Sec. (3.4)]{Dahl2004} and \citet[Eq.~(10)]{Biffis2005}, the $\mathcal{G}_t$-conditional density of $(\tau_x,\tau_y)$ on the set $\lbrace \tau_x>t,  \tau_y>t\rbrace$, denoted $f_t(t_1,t_2)$, is given by $f_t(t_1,t_2)  = \partial^2_{t_1t_2}\mathbb{Q}( (\tau_x>t_1)(\tau_y>t_2) | \mathcal{G}_t)$. 
%\begin{align}
%f_t(t_1,t_2)dt_1dt_2 &:=\mathbb{Q}( \tau_x\in [t_1, t_1+dt_1], \tau_y\in [t_2, t_2+dt_2]| \mathcal{G}_t)\nonumber \\
					 %&= \mathbb{E}^{\mathbb{Q}}\left[\mu_x(t_1)\mu_y(t_2) \left. e^{-\int_t^{t_1}\mu_x(s)ds}e^{-\int_t^{t_2}\mu_y(s)ds}\right\vert\mathcal{G}_t\right]dt_1dt_2.\label{EQ:MortalityDensity}
%\end{align}
Note that nothing so far was said on the intensities $(\mu_x(t),\mu_y(t))_{t \geq 0}$ apart from the fact that they have to be positive. If they are dependent, then there is a dependency between the mortality intensities even if conditionally on the intensities $(N_t^x)_{t \geq 0}$ and $(N_t^y)_{t \geq 0}$ are independent. It is through the intensities that we can incorporate a dependency between the mortality events $\tau_x$ and $\tau_y$.\\

Assuming $0 \leq t \leq T$, then the pricing at time $t$ of an insurance contingent claim payable at $T$ paying $1$ on the survival of annuitants aged $x$ and $y$ at time 0 is given on the set \(\lbrace \tau_x>t, \tau_y>t \rbrace \)  by
\begin{align}
\mathrm{SB}(t, T):=\mathbb{E}^{\mathbb{Q}}\left[\left. e^{-\int_t^T r(s)ds}\mathbf{1}_{\lbrace \tau_x>T \rbrace}\mathbf{1}_{\lbrace \tau_y>T \rbrace}\right\vert \mathcal{G}_t \right], \label{EQ:SB1}
\end{align}
where $(r(t))_{t\geq 0}$ is the risk free rate which is adapted to the filtration $\left(\mathcal{F}_t\right)_{t\geq 0}$. This instrument is commonly called the joint survival zero coupon bond. At this stage it is required to specify whether the risk free rate $(r(t))_{t\geq 0}$ depends on the mortality intensities $(\mu_x(t),\mu_y(t))_{t \geq 0}$. We follow the literature and assume that interest rates and the mortality risks are independent from each other, see for example \citet{Biffis2005}.\myfootnote{We acknowledge the strand in the SDE framework literature that correlates interest rates and mortality risk. See for example \citet{DeelstraGrasselliVanWeverberg2016} and \citet{DaFonseca2024} for works that use the Wishart process to handle that dependence.} \\

Furthermore, using \eqref{EQ:JointDeath} as well as the fact that the conditioning on $\mathcal{G}_t$ can be reduced to that on $\mathcal{F}_t$ as shown in \citet[Appendix C]{Biffis2005}, we obtain 
\begin{align}
\mathrm{SB}(t, T)=\mathbf{1}_{\lbrace \tau_x>t,\tau_y>t \rbrace}P(t,T)\mathbb{E}^{\mathbb{Q}}\left[\left. e^{-\int_t^T \mu_x(s) + \mu_y(s)ds}\right\vert \mathcal{F}_t \right], \label{EQ:SB2b}
\end{align}
%\begin{align}
%\mathrm{SB}(t, T)=P(t,T)\mathbb{E}^{\mathbb{Q}}\left[\left. \mathbf{1}_{\lbrace \tau_x>T \rbrace}\mathbf{1}_{\lbrace \tau_y>T \rbrace}\right\vert  \mathcal{G}_t \right], \label{EQ:SB2a}
%\end{align}
with $P(t,T)$ the time $t$-value of a zero-coupon bond with maturity $T$. Hereafter, $\mathbb{E}^{\mathbb{Q}}\left[\cdot| \mathcal{F}_t \right]$ will be denoted $\mathbb{E}_t^{\mathbb{Q}}\left[\cdot\right]$, so that the joint survival (zero-coupon) bond, denoted by $\mathrm{SB}(t,T)$, is eventually given on the set \(\lbrace \tau_x>t, \tau_y>t \rbrace \) by 
\begin{align}
\mathrm{SB}(t,T) &= P(t,T)\mathbb{E}_t^{\mathbb{Q}}\left[e^{-\int_t^T \mu_x(s)+ \mu_y(s) ds}\right] \nonumber\\
				 &= P(t,T)\mathrm{SB}_0(t,T), \label{EQ:SBZC}
\end{align}
where we use $\mathrm{SB}_0(\cdot, \cdot)$ to denote the joint survival bond that contains no interest rate risk. The equivalent actuarial notation is given by ${}_Tp_{xy}$,  however, due to the similarities of our work with \citet{Biffis2005} and \citet{DeelstraGrasselliVanWeverberg2016}, we instead follow their notation.\\

Using the framework introduced so far, the joint life annuity is therefore the sum of joint survival bonds, \textit{i.e.}, a product that pays at certain pre-specified dates an amount (possibly date dependent), conditional on the survival of both annuitants at those dates. Let $T_1,\ldots,T_N$ be a set of yearly spaced dates such that $T_N=\min(x^*-x-1,y^*-y-1)$ with $x^*$ and $y^*$ being the maximum attainable age for the annuitants $x$ and $y$ respectively. Consider the product that pays at times $T_i$ the amount $1$ conditional on the survival of both annuitants, the value of the product at time $t$ is given on the set \(\lbrace \tau_x>t, \tau_y>t \rbrace \) by
\begin{align}
\sum_{i=1}^N \mathrm{SB}(t,T_i), \label{EQ:SB3}
\end{align}
with $\mathrm{SB}(t,T)$ given by \eqref{EQ:SBZC}.\\

Another product of interest is the guaranteed joint survival annuity option, \textit{i.e.}, an option on \eqref{EQ:SB3}. We follow \citet[Section 3.2]{DeelstraGrasselliVanWeverberg2016}, who analyse the guaranteed survival annuity option (for a single annuitant), but consider the case of multiple annuitants. Suppose a guaranteed joint survival annuity option with maturity $T$ and consider a set of yearly spaced dates $T_1,\ldots, T_N$ such that $T_1=T + 1$ (where the units are years) and $T_N=\min(x^*-(x+T)-1,y^*-(y+T)-1)$ with $x^*$ and $y^*$ again being the maximum attainable age for annuitants $x$ and $y$ respectively. Then we denote the value of a guaranteed joint survival annuity option at time $t$, exercisable at $T$, on a joint survival annuity paying at times $T_1,...,T_N$ as $C(t,T,T_N)$. Thus at $t=T$, the payoff of the option is given on the set \(\lbrace \tau_x>T, \tau_y>T \rbrace \)  by
\begin{align}
C(t=T,T,T_N) &:= \max\left(g \sum_{i=1}^N \mathrm{SB}(T,T_i) ,1 \right) =1 + g\bar{C}(T,T,T_N), \label{EQ:Payoff}
		 %&=1+ g\left(\sum_{i=1}^N \mathrm{SB}(T,T_i) -1/g\right)_+ \nonumber \\ 
\end{align}
where $g$ is the fixed rate called the guaranteed joint survival annuity rate and $(z)_+=\max(z,0)$ and $\bar{C}(T,T,T_N) = \left(\sum_{i=1}^N \mathrm{SB}(T,T_i) -1/g\right)_+$. Therefore, the guaranteed joint survival annuity option value at time $t$ is given on the set \(\lbrace \tau_x>t, \tau_y>t \rbrace \) by
%\begin{align*}
%C(0,T,T_N) 	&= \mathbb{E}^{\mathbb{Q}}\left[ e^{-\int_{0}^{T}(r(s)+\mu_x(s)+\mu_y(s))ds} C(T,T,T_N) \right] \nonumber \\ 
			%&= P(0,T)\mathrm{SB}_0(0,T) + g \bar{C}(0,T,T_N),\nonumber 
%\end{align*}
\begin{align*}
C(t,T,T_N) 	&=P(t,T)\mathrm{SB}_0(t,T) + g \mathbb{E}_t^{\mathbb{Q}}\left[ e^{-\int_{t}^{T}(r(s) + \mu_x(s)+\mu_y(s))ds} \bar{C}(T,T,T_N) \right],\nonumber 
\end{align*}
with $\bar{C}(t,T,T_N) $ the expectation on the right hand side above.
%\begin{align}
%\bar{C}(0,T,T_N) := \mathbb{E}^{\mathbb{Q}}\left[ e^{-\int_{0}^{T}({\color{black}r(s)} + \mu_x(s)+\mu_y(s))ds} \bar{C}(T,T,T_N) \right].   \label{EQ:GAO} 
%\end{align}
It is worth having a closer look at the expression of $\bar{C}(t,T,T_N)$ as it highlights some modelling issues. Using \eqref{EQ:SBZC},  we obtain on the set \(\lbrace \tau_x>t, \tau_y>t \rbrace \) the equality
\begin{align}
\bar{C}(t,T,T_N) = \mathbb{E}_t^{\mathbb{Q}}\left[ e^{-\int_{t}^{T}(r(s) + \mu_x(s)+\mu_y(s))ds}\left( \sum_{i=1}^NP(T,T_i)\mathrm{SB}_0(T,T_i)-1/g \right)_+  \right],  \label{EQ:GAO2} 
\end{align}
where we can clearly see that even if the interest rate is independent with respect to the mortality intensities, its distribution impacts the option price. Given the difficulty of treating stochastic interest rates, independent or dependent of the mortality intensities, we take the interest rate to be deterministic.\\
% We elaborate further on the treatment of stochastic interest rates in our model in Section~\ref{Subsection:interest_rates}.\\

Let us stress the fact that all of the equations stated above are model independent, and are not contingent on using any specific model or process. \\

As seen from \eqref{EQ:GAO2}, the pricing of a guaranteed (joint) annuity option is a difficult task, as it requires the law of the sum of joint survival (zero coupon) bonds. If the joint survival bond is an exponential function of the state variable, then this amounts to requiring the law of the sum of exponentials of random variables that is in general unknown. This difficulty is well explained in \citet{BiffisMillossovich2006} and is one of the main difficulties we need to overcome.

%%%%%%%%%%%%%%%%%%%%%%%%%%%%%%%%%%%%%%%%%%%%%%%%%%%%%%%%%%
%
% The model
%
%%%%%%%%%%%%%%%%%%%%%%%%%%%%%%%%%%%%%%%%%%%%%%%%%%%%%%%%%%
\section{The linear-rational Wishart model}\label{LRWishart}

To build the joint mortality model we use the Wishart process defined in \citet{Bru1991} and introduce its main properties. In the following section, we provide a concise review of the essential techniques and findings required to construct our model. This presentation is prompted by the relatively novel application of the Wishart process in mortality modeling. Given a filtered probability space $(\Omega,\mathcal{F}, (\mathcal{F}_t)_t, \mathbb{P})$ we denote by $\mathbb{E}\left[\,\cdot\,\right]$ (resp. $\mathbb{E}_t\left[\,\cdot\,\right]:=\mathbb{E}\left[\,\cdot\,|\mathcal{F}_t \right]$) the expectation (resp. conditional expectation) under the historical probability measure $\mathbb{P}$. The Wishart process satisfies the matrix SDE
\begin{equation}
dv_t = (\omega + m v_t + v_t m^\top )dt + \sqrt{v_t}dw_t \sigma +  \sigma^\top dw_t^\top \sqrt{v_t}\,, \label{EQ:Wishart}
\end{equation}
where $v_t$ is an $n\times n$ matrix that belongs to the set of positive definite matrices denoted $\mathbb{S}_n^{++}$, $m,\sigma$ belong to the set of $n\times n$ real  matrices denoted $\mathsf{M}(n)$, $\lbrace w_t;t\geq 0 \rbrace$ is a matrix Brownian motion of dimension $n\times n$ (\textit{i.e.}, a matrix of $n^2$ independent scalar Brownian motions) under the probability measure $\mathbb{P}$ and $\cdot^\top$ stands for the matrix transposition. The matrix $\omega \in \mathbb{S}_n^{++}$ satisfies certain constraints involving $\sigma^{\top}\sigma$ to ensure the positive-definiteness of the matrix process $v_t$. The quantity $\sqrt{v_t}$ is well defined since $v_t \in \mathbb{S}_n^{++}$. 
The matrix $m$ is such that $\lbrace \Re(\lambda_i^{m})<0; i=1,\ldots,n \rbrace $ where $\lambda_i^{m}\in \mathsf{Spec}(m)$ for $i=1,\ldots,n$ and $\mathsf{Spec}(m)$ is the spectrum of the matrix $m$, while $\Re(\,\cdot\,)$ denotes the real component.
 The matrix $\sigma$ belongs to $\mathsf{GL}_n(\mathbb{R})$ the general linear group over $\mathbb{R}$ (\textit{i.e.}, the set of real invertible matrices). Due to the invariance of the law of the Brownian motion with respect to rotations and the polar decomposition of $\sigma$, we can assume that $\sigma \in \mathbb{S}_n^{++}$. We denote $e_{ij}$ to be the basis of $\mathsf{M}(n)$, \textit{i.e.}, it is the $n\times n$ matrix with $1$ in the $(i,j)$ place and zero elsewhere. Lastly $I_n$ stands for the $n\times n$ identity matrix, $0_n$ is the $n\times n$ null matrix while $\textup{diag}(z)$ with $z\in \mathbb{R}^n$ is the $n\times n$ matrix with $z$ on its diagonal. The infinitesimal generator of the Wishart process is given by \citet{Bru1991}:
\begin{align}
\mathcal{L} = \textup{tr}[(\omega + m v + vm^\top)D + 2 v D \sigma^2 D]\,, \label{EQ:InfinitesimalGenerator}
\end{align}
where $\textup{tr}[\,\cdot\,]$ is the trace of a matrix and $D$ is the $n\times n$ matrix operator $D_{ij}:=\partial_{v_{ij}}$.\\

The Wishart process was initially defined and analysed in \citet{Bru1991} under the assumption that 
\begin{equation}
	\omega=\beta \sigma^2 \label{EQ:BruOmegaSpecification}
\end{equation}
with $\beta \in \mathbb{R}_+$ such that $\beta \geq n+1$  to ensure that $v_t \in \mathbb{S}_n^{++}$. Hereafter, this specification will be referred to as the Bru case. It was later extended in \citet{MayerhoferPfaffelStelzer2011} \citep[see also][]{CuchieroFilipovicMayerhoferTeichmann2011} to the case $\omega \in \mathbb{S}_n^{++}$ and proved that if 
\begin{align}
	\omega - \beta \sigma^2 \in \mathbb{S}_n^{++}, \label{EQ:PositiveConstraint}
\end{align}
with $\beta\geq n+1$ then $v_t\in \mathbb{S}_n^{++}$.\\

In our work, we deliberately restrict ourselves to the Bru case, (\textit{i.e.}, $\omega=\beta \sigma^2$).
The advantage of the Bru case is that moment generating function can be explicitly integrated making it fully explicit. This allows us to derive analytic formulae surrounding the guaranteed joint survival annuity option and the individual mortalities in contrast to the work of \cite{DaFonseca2024} which is not in the Bru case. The below result is well known. We refer to \citet{Bru1991} or \citet{GrasselliTebaldi2008} \citep[see also][Chap. 5]{Alfonsi2015}.\myfootnote{ Following the introduction of the Wishart process in finance by \citet{GourierouxSufana2010}, applications have been developed in equity and foreign exchange derivatives (see \citet{DaFonsecaGrasselliTebaldi2007}, \citet{LeungWongNg2013}, \citet{GnoattoGrasselli2014}, \citet{LaBuaMarazzina2022}, and \citet{FarazBehzadHusseinArianEscobarAnel2025}); in interest rate derivatives (see \citet{GourierouxSufana2011}, \citet{Gnoatto2012}, and \citet{CuchieroFontanaGnoatto2019}); and in optimal portfolio choice (see \citet{BauerleLi2013} and \citet{BrangerMuckSeifriedWeisheit2017}), to name just a few works across these fields. In all these studies, the general dependency structure among the components of the Wishart process, along with its strong analytical properties, plays a crucial modelling role.}

\begin{proposition}\label{Prop:MGFBru}
	Suppose that $\omega$ in  \eqref{EQ:Wishart} is such that $\omega=\beta \sigma^2$ with $\beta \in \mathbb{R}$ and $\beta \geq n+1$ then the moment generating function of $v_t$ for $\theta_1\in \mathbb{S}_n$ is given by
	\begin{align}
		\Phi(t,\theta_1,v_0) &=\mathbb{E}\left[ \exp\left( \textup{tr}[\theta_1 v_t]  \right)\right]\nonumber \\
		&=\frac{\textup{etr}\left(\frac{\vartheta_t^\top}{2}(2\varsigma_t \theta_1)(I_n-2\varsigma_t\theta_1)^{-1}  \right)}{\det(I_n -2\varsigma_t\theta_1)^{\beta/2}}, \label{EQ:mgfBru}
	\end{align}
	with $\varsigma_t :=\int_0^te^{(t-s)m}\sigma^2 e^{(t-s)m^\top}ds$, which is given by $\textup{vec}(\varsigma_t)=\mathsf{A}^{-1}(e^{\mathsf{A}t}-I_{n^2})\textup{vec}(\sigma^2)$,  $\vartheta_t := \varsigma_t^{-1}e^{mt}v_0e^{m^\top t}$, $\textup{etr}(\cdot)=\exp(\textup{tr}[\cdot])$ and $\det(\cdot)$ the determinant of a matrix. Here $\textup{vec}\left(\,\cdot\,\right)$ denotes the vec operator that transforms a $n\times n$ matrix into a $n^2$ vector by stacking the columns, $\mathsf{A}:=I_{n}\otimes m+m\otimes I_{n}$, and $\otimes$ denotes the Kronecker product.
\end{proposition}

The Wishart process is affine, that is the moment generating function is exponentially affine in the state variable. The equation \eqref{EQ:mgfBru} is the moment generating function of a non-central Wishart distribution, hence the name of the process, and the corresponding density is given in \citet[Eq.~3.5.1 p.~114]{GuptaNagar2000}. Note that even if the density is known, computing an expectation requires to perform a $n(n+1)/2$ dimensional numerical integration.

\begin{remark}
	Let us stress the fact that when $\omega \neq \beta \sigma^2$, then the moment generating function of $v_t$ is unknown, and in particular it is not a Wishart distribution.  In this case the name of the process is misleading. The analytic results derived in Section~\ref{PricingAnnuity} and Section~\ref{PricingAnnuityDerivative} makes use of this explicit moment generating function, which are not possible to be derived in the setting of \cite{DaFonseca2024}.
\end{remark}

The results developed in this work rely on the following result, which arises from solving Lyapunov equations after rewriting the dynamics of the Wishart process. This fact is well known in the matrix Riccati literature, see for example \cite{Abou‐Kandil2003MatrixRiccati}.

\begin{lemma}\label{LEM:TraceMean}
Let $v_t$ be a Wishart process described by \eqref{EQ:Wishart}. If $u_0 \in \mathsf{M}(n)$, then there exist a function $b_0(t)$ and a matrix function $a_0(t) \in \mathsf{M}(n)$ such that
\begin{align}
\mathbb{E}\left[\textup{tr}[u_0v_{t}]\right]&=\textup{tr}[a_0(t)v_{0}] + b_0(t), \label{EQ:OdeMeanWishartIntegratedMatrix}
\end{align}
with
\begin{align}
	\textup{vec}\left(a_0(t)^{\top}\right) &= e^{\mathsf{A}^{\top}t}\textup{vec}(u_0^\top),\\
  b_0(t) &= \textup{vec}(u_0^\top)^\top \mathsf{A}^{-1}\left(e^{\mathsf{A}t}-I_{n^2} \right)\mathsf{b}. 
\end{align}
Here $\mathsf{b}:=\textup{vec}\left(\omega\right)$, $\mathsf{A}:=I_{n}\otimes m+m\otimes I_{n}$, and $\otimes$ denotes the Kronecker product.

\end{lemma}

%If the process \eqref{EQ:Wishart} is stationary, which requires the eigenvalues of $m$ to be negative and if they are then the eigenvalues of $\mathsf{A}$ are also negative following the Kronecker product property, it leads to $\bar{v}_{\infty}:=\lim_{t\to+\infty}\mathbb{E}\left[v_{t}\right]$
%that solves the equation 
%\begin{align*}
%	\textup{vec}\left(\bar{v}_{\infty}\right)=-\mathsf{A}^{-1}\mathsf{b}\,,
%\end{align*}
%which corresponds to  the matrix equation
%\begin{align}
%m\bar{v}_{\infty}  + \bar{v}_{\infty} m^\top = -\omega. \label{EQ:Xinfty}
%\end{align}

%\begin{remark}\label{Rem:Trace}
%Notice that if $a\in \mathbb{S}_n$ and $b \in \mathsf{M}(n)$ then $\textup{tr}[ab]=\textup{tr}[a(b+b^\top)]/2$.
%\end{remark}

The following proposition is well known in the literature and is useful for obtaining the quadratic variation of linear functions of the Wishart process, which in turn will be utilised to derive the covariation (or instantaneous correlation) between mortality intensities. A proof of the proposition appears in p. 214 of \cite{DeelstraGrasselliVanWeverberg2016} and reads as follows.
\begin{proposition}\label{Prop:QuadraticVariation}
Given  $h_1, h_2\in \mathsf{M}(n)$ the quadratic variation between $(\textup{tr}[h_1v_t])_{t \geq 0 }$ and $(\textup{tr}[h_2v_t])_{t \geq 0 }$ is  
\begin{align*}
d\langle \textup{tr}[h_1v_.],\textup{tr}[h_2v_.]\rangle_t&=\textup{tr}[(h_1+h_1^\top) v_t(h_2+h_2^\top)\sigma^2] dt\,.
\end{align*}
\end{proposition}
When $n=2$, using Proposition \ref{Prop:QuadraticVariation}, one can in particular show that the quadratic covariation between $v_{11,t}$ and $v_{22,t}$ is given by
\begin{align}
d\langle v_{11,.},v_{22,.}\rangle_t &= 4v_{12,t}(\sigma^2)_{12}dt \label{EQ:v11v22},
\end{align}  
where $(\sigma^2)_{ij}$ is the $(i,j)$ element of the matrix $\sigma^2$. Note that $(\sigma^2)_{12}=0$ if $\sigma_{12}=0$, which means the dependence between $v_{11,t}$ and $v_{22,t}$ is primarily dependent on $\sigma_{12}$.\\

To obtain the derivative of the moment generating function with respect to a parameter, the next proposition provides the details in the Bru case. This result will be particularly helpful for obtaining the densities of mortality intensities, as well as risk measures such as expected shortfall and value at risk for joint survival annuities.
\begin{proposition}\label{Prop:MGFBruDerivative}
Suppose that $\omega$ in  \eqref{EQ:Wishart} is such that $\omega=\beta \sigma^2$ with $\beta \in \mathbb{R}$ and $\beta \geq n+1$ and consider the moment generating function of $v_t$ given by \eqref{EQ:mgfBru}. For $\theta_1\in \mathbb{S}_n$ and $\theta_2\in \mathbb{S}_n$ consider the function for $\nu \in \mathbb{R}$, $ \nu \to \Phi(t,\theta_1 + \nu \theta_2,v_0) = \mathbb{E}\left[\textup{etr}\left( (\theta_1+\nu \theta_2) v_t  \right)\right]$ then its derivative with respect to $\nu$  is given by
\begin{align}
\Phi_\nu(t,\theta_1,\theta_2,\nu, v_0)	&:=\frac{d\Phi(t,\theta_1+\nu \theta_2, v_0)}{d\nu}\nonumber \\
											& =(g_1+g_2+g_3)\Phi(t,\theta_1+\nu \theta_2, v_0), \label{EQ:mgfBruDerivative}
\end{align}

with 
\begin{align*}
g_1(t,\theta_1,\theta_2,\nu) &= \textup{tr}\left[\frac{\vartheta_t^\top}{2} 2\varsigma_t\theta_2(I_n -2\varsigma_t(\theta_1+\nu \theta_2))^{-1} \right], \\
g_2(t,\theta_1,\theta_2,\nu) &= \textup{tr}\left[\frac{\vartheta_t^\top}{2} 2\varsigma_t(\theta_1+ \nu \theta_2)(I_n -2\varsigma_t(\theta_1+\nu \theta_2))^{-1}2\varsigma_t \theta_2(I_n -2\varsigma_t(\theta_1+\nu \theta_2))^{-1}\right], \\
g_3(t,\theta_1,\theta_2,\nu) &= \frac{\beta}{2}\textup{tr}\left[2\varsigma_t\theta_2(I_n -2\varsigma_t(\theta_1+\nu \theta_2))^{-1} \right],
\end{align*}
and $\varsigma_t$ and $\vartheta_t$ as in Proposition \ref{Prop:MGFBru}.
\end{proposition}

%%%%%%%%%%%%%%%%%%%%%%%%%%%%%%%%%%%%%%%%%%%%%%%%%%%%%%%%%%%%%%%%%%%%%%%%%%%%%%%%%%%%%%%%%%%%%%%%%%%%%%%%%%%%%%%%%%%%%%%%%%%%%%%%%%%%%%
%
% 
%
%%%%%%%%%%%%%%%%%%%%%%%%%%%%%%%%%%%%%%%%%%%%%%%%%%%%%%%%%%%%%%%%%%%%%%%%%%%%%%%%%%%%%%%%%%%%%%%%%%%%%%%%%%%%%%%%%%%%%%%%%%%%%%%%%%%%%%

We build on the potential approach proposed by \citet{Rogers1997} whose main idea is to directly specify the state-price density in \eqref{EQ:SBZC} so that the corresponding instantaneous mortality intensities are positive.\myfootnote{The potential approach of \citet{Rogers1997} is well documented in standard textbooks (see, for example, \citet{Bjorck2009} and \citet{Cairns2004book}) and has been used in several interest rate derivatives studies (see, for example, \citet{CrepeyMacrinaNguyenSkovmand2015}, \citet{NguyenSeifried2015}, { \citet{FilipovicLarssonTrolle2017}} or \citet{NguyenSeifried2021}). When combined with a Wishart process, it has proved particularly effective in the multi-curve setting of \citet{DaFonsecaDawuiMalevergne2022} and in a cross-asset context in \citet{DaFonseca2024}. We believe the framework still offers many modelling possibilities, including for mortality derivatives and credit risk derivatives. {For example, \citet{DeGiovanniPirraViviano2025} recently proposed a model that departs from the traditional modelling approach of \citet{Dahl2004} and \citet{Biffis2005}. However, their model differs from ours primarily in the way the state-price density is defined.}  } The standard approach proceeds the other way around, \textit{i.e.,} to specify the mortality intensities and deduce the state-price density.  We briefly summarise the mathematical framework as well as develop our main results on mortality modelling.\\

First, we define the state-price density as
\begin{align}
\zeta_t&:= e^{-\alpha t} (1+ \textup{tr}[u_0 v_t])\,, \label{EQ:PricingKernel}
\end{align}
with $\alpha \in \mathbb{R}_+$,  $u_0 = u_1 + u_2$ with $u_1 \in \mathbb{S}_n^{+}$ (the set of positive semi-definite matrices) and $u_2 \in \mathbb{S}_n^{+}$ (so that $u_0 \in \mathbb{S}_n^{+}$) and $(v_t)_{t \geq 0}$ a Wishart process.  The state-price density can be rewritten as follows. Define the positive function $f: \mathbb{S}_n^{++} \to \mathbb{R}^+$ such that $f(v) := 1+\textup{tr}[u_0v]$ then the state-price density writes as $\zeta_t=e^{-\alpha t}f(v_t)$ and is positive as $f$ is. Define  $g(v):=(\alpha - \mathcal{L})f(v)$ and assume that it is a positive function for sufficiently large $\alpha$. Using \citet[Eqs.~(1.2),(1.3),(1.4)]{Rogers1997}, the state-price density allows us to compute $\mathrm{SB}_{0}(t,T)$ given by \eqref{EQ:SBZC}, or any other mortality related product, if we assume the following change of probability measure
\begin{align}
\zeta_t = e^{-\int_0^t (\mu_x(s) + \mu_y(s))ds }\left.\frac{d\mathbb{Q}}{d\mathbb{P}}\right\vert_{\mathcal{F}_t}. \label{EQ:Radon}
\end{align}

Indeed, suppose that one needs to compute the value at time $t$ of $\Pi_t$ that is equal to $1$ at time $T$ conditional on the annuitants $x$ and $y$ being alive at that time, then its value is given by
\begin{align*}
\Pi_t 	&= \mathbb{E}_t^{\mathbb{Q}}\left[e^{-\int_t^T (\mu_x(s) + \mu_y(s))ds } \right],\\
		&= \mathbb{E}_t\left[\frac{\zeta_T}{\zeta_t}\right],
\end{align*} 
which is known  explicitly for a suitable choice of the state-price density $(\zeta_t)_{t\geq 0}$.\myfootnote{Note that the pricing is done under the historical probability measure. This contrasts with the standard approach, where the dynamics of the state variable are specified under the risk-neutral measure. See \citet{DaFonseca2024} for some implications.} The following section shows how the expectations presented in Section~\ref{Framework} can be handled.\footnote{ Since the interest rates are deterministic, the pricing kernel is about the mortality risk.}\\

Our choice of the positive function \( f(v) = 1 + \textup{tr}[u_0v] \) is a deliberate aspect of the model specification. The potential approach proposed by \cite{Rogers1997} necessitates \( f \) to be a positive function of the state variable. We selected this particular form of \( f \) to ensure that the joint survival annuity price remains a linear-rational function of the state variable \( v_t \). It is important to note that alternative choices for \( f \) could impact the results discussed later. However, we leave the exploration of these alternatives for future research.

\begin{remark}\label{REM:alt_f}
	The results in this work rely upon the linear-rational structure to obtain the price of the joint survival annuity and the guaranteed joint survival annuity option. It is straightforward to see the set of admissible $f: \mathbb{S}_n^{+} \to \mathbb{R}^+$ that retains the linear-rational structure must take the form
			\[
			f(v) = \phi + \tr[\theta_1 v]
			\]
			where $\phi \in \mathbb{R}$ and $\theta_1$ is a positive semi-definite matrix. Since $\phi$ shifts the level of $f(v)$, it could be difficult to disentangle with the parameters of the Wishart process that control its level. To avoid identification issues, we set $\phi = 1$ and adopt the specification  $f(v) := 1 + \textup{tr}[u_0v]$.
\end{remark}

%%%%%%%%%%%%%%%%%%%%%%%%%%%%%%%%%%%%%%%%%%%%%%%%%%%%%%%%%%
%
% The joint survival annuity and related quantities.
%
%%%%%%%%%%%%%%%%%%%%%%%%%%%%%%%%%%%%%%%%%%%%%%%%%%%%%%%%%%
\section{Joint survival annuity valuation}\label{PricingAnnuity}

To illustrate the methodology on the linear-rational Wishart mortality model, let us start with the joint survival bond that can be explicitly computed as the following proposition shows. 
\begin{proposition}\label{Prop:SBZCPricing}
Let $u_1 \in \mathbb{S}_n^+$ and $u_2 \in \mathbb{S}_n^+$ and define $u_0 = u_1 + u_2$. Assume further that \eqref{EQ:PricingKernel} holds. Then on the set \(\lbrace \tau_x>t, \tau_y>t \rbrace \) the joint survival bond $\mathrm{SB}(t,T)$ given by \eqref{EQ:SBZC} is known explicitly as we have
\begin{align}
\mathrm{SB}(t,T)&=P(t,T) e^{-\alpha(T-t)}\frac{1 + b_0(T-t) +\textup{tr}[a_0(T-t) v_{t}]}{1+\textup{tr}[u_0 v_{t}]}, \label{EQ:SBZC3}
\end{align}
where $ b_0(t):=b_1(t)+b_2(t)$ and $a_0(t):=a_1(t)+a_2(t)$ with $(b_1(t),a_1(t))$ and $(b_2(t),a_2(t))$ two sets of functions obtained using Lemma \ref{LEM:TraceMean} for $u_1$ and $u_2$, respectively.
\end{proposition}
We should note that the above proposition is identical to the result derived in \cite{DaFonseca2024}. This is to be expected given the result is independent of the parameterisation and the modelling choices we make on the Wishart process.\myfootnote{Note that any symmetric positive-definite matrix-valued process sharing the Wishart drift, regardless of its diffusion term, yields the same survival-bond expression, because the expression depends only on the drift of the process.} From now on, we will avoid mentioning that the expectations conditional on a given time \(t\) are only valid on the sets \(\lbrace \tau_x>t, \tau_y>t \rbrace \) or \(\lbrace \tau_x>t\rbrace \), depending on the expectation considered.
 
According to \citet[Eq.~(1.5)]{Rogers1997} or \citet[Eq. (28.42)]{Bjorck2009} and \citet[p. 134]{Cairns2004book}, we have
\begin{align*} 
\mu_x(s)+\mu_y(s) 	&= \frac{(\alpha - \mathcal{L})f}{f}= \frac{\alpha +\alpha \textup{tr}[u_0v_s]-\textup{tr}[u_0\omega] - 2 \textup{tr}[u_0mv_s]}{1+\textup{tr}[u_0v_s]}.
\end{align*} 
Through our judicious choice of $f$, we are able to isolate $\alpha$. Below we define the mortality intensities to be
\begin{align} 
\mu_x(s) 		&= \frac{\alpha/2 +\alpha \textup{tr}[u_1v_s]-\textup{tr}[u_1\omega] - 2 \textup{tr}[u_1mv_s]}{1+\textup{tr}[u_0v_s]}, \label{EQ:IntensityX}\\
\mu_y(s) 		&= \frac{\alpha/2 +\alpha \textup{tr}[u_2v_s]-\textup{tr}[u_2\omega] - 2 \textup{tr}[u_2mv_s]}{1+\textup{tr}[u_0v_s]}, \label{EQ:IntensityY}
\end{align} 
where we have chosen to apportion the term $\alpha$ equally to both annuitants.\myfootnote{These mortality intensities are structurally the same to the short rate / mortality intensity (35)-(36) in \cite{DaFonseca2024}. Alternative specifications would involve either changing $f$ as mentioned in Remark~\ref{REM:alt_f}, or on how the intensities are defined in \eqref{EQ:IntensityX}-\eqref{EQ:IntensityY}, or through the choices of $u_1, u_2$.}
From now on, we suppose that there exists $\alpha$ such that $\mu_x(s)$ given by \eqref{EQ:IntensityX} and $\mu_y(s)$ given by \eqref{EQ:IntensityY} are positive. The existence of an $\alpha$ such that the intensities are positive is dependent on the choice of the matrices $u_1, u_2$ and the parameters of the Wishart process. We provide a sufficient requirement in the case when $n=2$ subsequently.

\begin{remark}
Let us reiterate the fact that the potential approach of \citet{Rogers1997} proceeds differently from the traditional mortality modelling approaches of \citet{Dahl2004} and \citet{Biffis2005}, which consist of specifying the dynamics of the mortality intensity, thereby determining the state-price density. In contrast, \citet{Rogers1997} specifies the state-price density, which in turn determines the intensity. Conditions on the state-price density ensure that the intensity remains positive. Note also that, since the dynamics \eqref{EQ:Wishart} of \((v_t)_{t \geq 0}\) are given under \(\mathbb{P}\), the dynamics of the intensities are also specified under this probability measure, and therefore pricing is conducted under \(\mathbb{P}\). Following \citet[p. 164]{Rogers1997}, it is possible to determine the dynamics of \((v_t)_{t \geq 0}\) under the risk-neutral probability measure, and thus also the dynamics of the intensities (see \citet{DaFonseca2024} for an example applied to the Wishart process). %Applying this result to our problem leads for the state variable an infinitesimal generator \(\mathcal{L}^*\) of the form
%\begin{align}
		%\mathcal{L}^*f(v) = \mathcal{L}f(v) + \frac{\textup{tr}[2\sigma^2u_0vDf] + \textup{tr}[2vu_0\sigma^2Df]}{1+\textup{tr}[u_0v]}\,, \label{EQ:infinite_gen_Q_detail}
	%\end{align}
	%with \(\mathcal{L}\) given by \eqref{EQ:InfinitesimalGenerator}, so that the dynamic of $(v_t)_{t\geq 0}$ under $\mathbb{Q}$ is known.
	However, since these are not needed for our purposes, we omit the details.
\end{remark}

\begin{remark}
 Note, we can easily extend our model to the case of an $k$-joint annuity as follows. First, we should have a dimension of the Wishart process to be at least $n \geq k$. We then set \(u_0 = \sum_{i=1}^k u_i\), and similar to the above intensities in \eqref{EQ:IntensityX} and \eqref{EQ:IntensityY}, we could define the $i$'th annuitant's mortality intensity as
\[
	\mu_i(s) = \frac{\alpha/k +\alpha \textup{tr}[u_iv_s]-\textup{tr}[u_i\omega] - 2 \textup{tr}[u_imv_s]}{1+\textup{tr}[u_0v_s]}.
\]
\end{remark}

In order to understand how the positivity enters into the problem let us consider the case of $n=2$ (with noting that $n$ could possibly be larger, and thus be a multi-factor mortality model). By assumption $u_1\in \mathbb{S}_n^{+}$ and $u_2\in \mathbb{S}_n^{+}$ so that $\textup{tr}[u_1\omega]$, $\textup{tr}[u_2\omega]$, $\textup{tr}[u_1v_s]$ and  $\textup{tr}[u_2v_s]$ are  positive as $\omega \in \mathbb{S}_n^{++}$ while $v_s$ belongs to $\mathbb{S}_n^{++}$. If $\textup{tr}[u_1mv_s]<0$ and $\textup{tr}[u_2mv_s]<0$ then as long as $\alpha/2>\max(\textup{tr}[u_1\omega],\textup{tr}[u_2\omega])$ then $\mu_x(s)$ and $\mu_y(s)$ are positive. To gain a little bit more of understanding of the positivity problem, it is worth considering the simple case $n=2$, $m$ diagonal, $u_1=e_{11}$ and $u_2=e_{22}$ then in that particular case $\mu_x(s) $ and $\mu_y(s) $ rewrite as
\begin{align} 
\mu_x(s)			&= \frac{\alpha/2 +\alpha v_{11,s}-\omega_{11} - 2 m_{11}v_{11,s}}{1 + v_{11,s} + v_{22,s}}, \label{EQ:mu_x_intensity_simple} \\
\mu_y(s) 			&= \frac{\alpha/2 +\alpha v_{22,s}-\omega_{22} - 2 m_{22}v_{22,s}}{1 + v_{11,s} + v_{22,s}}, \label{EQ:mu_y_intensity_simple}
\end{align} 
where the condition $\alpha/2>\max(\omega_{11},\omega_{22})$ clearly shows that it is sufficient to ensure the positivity of both $\mu_x(s)$ and $\mu_y(s)$ as the matrix $m$ has negative eigenvalues (which means the diagonal terms are negative since $m$ is diagonal). Furthermore, the correlation between these two variables is driven by $d\langle v_{11,.},v_{22,.}\rangle_s$ given by \eqref{EQ:v11v22}, which means it is controlled by the coefficient $(\sigma^2)_{12}$  that can take any sign and is zero if $\sigma_{12}=0$. As such, $\sigma_{12}$ is the coefficient through which we control the dependency between the two intensities. Analysing the scaled mortality intensities $(1+\textup{tr}[u_0v_s])\mu_x(s)$ and $(1+\textup{tr}[u_0v_s])\mu_y(s)$, we can find their instantaneous quadratic covariation which is given by  $d\langle \textup{tr}[h_1v_.],\textup{tr}[h_2v_.]\rangle_s$ with $h_1=\alpha u_1-2u_1m$ and $h_2=\alpha u_2-2u_2m$.  Using Proposition~\ref{Prop:QuadraticVariation} we get 
\begin{align}
	d\langle (1+\textup{tr}[u_0v_.])\mu_x(.),(1+\textup{tr}[u_0v_.])\mu_y(.)\rangle_s &= d\langle \textup{tr}[h_1v_.],\textup{tr}[h_2v_.]\rangle_s \nonumber\\
	&= \textup{tr}[(h_1+h_1^\top) v_s(h_2+h_2^\top)\sigma^2] ds\,, \label{EQ:intensity_covariation}
\end{align}
which demonstrates the impact of the parameters on the covariance between the mortality intensities. Note that when $u_1 = e_{11}$ and $u_2 = e_{22}$ then the covariation given by \eqref{EQ:intensity_covariation} will be zero if $\sigma_{12}=0$.\myfootnote{Note that while the scaled intensities will have a quadratic covariation of zero if $\sigma_{12}=0$, the quadratic covariation between the intensities will not zero due to the common denominators in \eqref{EQ:IntensityX} and \eqref{EQ:IntensityY}.} Finally, analysing the mortality intensities in \eqref{EQ:mu_x_intensity_simple} and \eqref{EQ:mu_y_intensity_simple}, given that $\alpha > 0$, \(\omega_{ii}>0\), and \(m_{ii} < 0\), we find \(\frac{\partial \mu_x(s)}{\partial v_{11}} > 0\) and \(\frac{\partial \mu_x(s)}{\partial v_{22}} > 0\). This provides a natural link between the Wishart states and mortality: as $v_{11,s}$ and $v_{22,s}$ increase, the mortality intensities of $x$ and $y$ also increase, respectively. Thus, the Wishart process can replicate an increasing mortality trend over time in expectation, while remaining stationary.\\

The expression \eqref{EQ:IntensityX} combined with the analytical results of Propositions \ref{Prop:MGFBru} as well as the results in \citet{Gurland1948} enables the computation of the mortality density as the following proposition shows.\\
\begin{proposition}\label{Prop:MortalityIntensityDistribution}
Let $\mu_x(T)$ be the mortality intensity of the annuitant $x$ that is given by \eqref{EQ:IntensityX} and denote $c_1=\alpha/2 -\textup{tr}[u_1\omega]$, $h_1 = \alpha u_1-2u_1m$, then the cumulative distribution function (conditional to time \(t\)) of $\mu_x(T)$  is 
\begin{align}
\mathbb{P}(\mu_x(T) \leq z) = \frac{1}{2}-\frac{1}{\pi}\int_0^{+\infty}\Im \left(\frac{e^{\mathrm{i} s(c_1-z)}\Phi(T-t,\mathrm{i} s(h_1-zu_0),v_t)}{s}\right) ds,\label{EQ:IntensityDistribution}
\end{align}
where $\Im(.)$ stands for the imaginary part of a complex number and $\Phi(.,.,.)$ is given in Proposition \ref{Prop:MGFBru}.  
\end{proposition}
From the above proposition, we can derive the mortality intensity density by simply taking the derivative of the cumulative distribution, which confirms  the usefulness of Proposition  \ref{Prop:MGFBruDerivative}, as the following result shows.
\begin{corollary}\label{Prop:MortalityIntensityDensity}
Let $\mu_x(T)$ be the mortality intensity of the annuitant $x$ that is given by \eqref{EQ:IntensityX}, and thus its cumulative distribution function is given by Proposition \ref{Prop:MortalityIntensityDistribution}, then the density (conditional to time t) of $\mu_x(T)$ is given by
\begin{align}
\frac{d}{dz} \mathbb{P}(\mu_x(T) \leq z) &= -\frac{1}{\pi}\int_0^{+\infty}\Im \left(\frac{(-\mathrm{i}s + g_1+g_2+g_3)e^{\mathrm{i}s(c_1-z)}\Phi(T-t,\mathrm{i}s(h_1-zu_0),v_t)}{s} \right)ds,\label{EQ:IntensityDensity}
\end{align}  
with $g_j=g_j(T-t,\mathrm{i} sh_1,-\mathrm{i} su_0,z)$ for $j\in \lbrace 1,2,3 \rbrace$ as in Proposition \ref{Prop:MGFBruDerivative}.
\end{corollary}

%\begin{remark}
%When it comes to model implementation, \eqref{EQ:IntensityDistribution} may require to specify the limit at zero of the integrand. Performing a Taylor expansion (of the exponential function) one gets
%\begin{align*}
%\lim_{s\to 0}\Im \left(\frac{e^{\mathrm{i} s(c_1-z)}\Phi(t,\mathrm{i} s(h_1-zu_0),v_0)}{s}\right)=c_1-z + \mathbb{E}\left[\textup{tr}[(h_1-zu_0)v_t] \right],
%\end{align*}
%that is computed thanks to Lemma \ref{LEM:TraceMean}, while for \eqref{EQ:IntensityDensity}, deriving that function leads to $-1 - \mathbb{E}\left[\textup{tr}[u_0v_t]\right]$.
%\end{remark}

\begin{remark}
	 Note that the results above require an explicit expression for the moment generating function, which is not available under the assumptions made in \cite{DaFonseca2024}. 
\end{remark}

Having an explicit form of the mortality density is of significant practical importance. For example, one can compute the conditional means or higher order moments of the mortality intensities, which are modelled under the historical measure $\mathbb{P}$, that can guide in the operational use of the model.\\

Thanks to \eqref{EQ:SBZC3} the joint survival bond $\mathrm{SB}(t,T)$ (\textit{i.e.}, \eqref{EQ:SBZC}) is known and therefore the joint survival annuity \eqref{EQ:SB3} is also known as the next proposition shows.
\begin{proposition}\label{Prop:JointSurvivalAnnuity} 
Let $T_1<\ldots<T_N$ be a set of yearly spaced dates such that $t<T<T_1$ and $T_N=\min(x^*-(x+T+ t)-1,y^*-(y+T+t)-1)$ with $x^*$, resp. $y^*$, the maximum age the $x$, resp. $y$, insured can reach. The joint survival annuity pays 1 conditional on the survival of both annuitants at those dates, its value at time $T$ is given by
\begin{align}
\sum_{i=1}^N \mathrm{SB}(T,T_i)=\frac{b_3(T,T_N)+\textup{tr}[a_3(T,T_N)v_T]}{1 + \textup{tr}[u_0v_T]}, \label{EQ:JointSurvivalAnnuity}
\end{align}
with
\begin{align}
b_3(T,T_N)  &:= \sum_{i=1}^N P(T,T_i)e^{-\alpha(T_i-T)} \left( 1 +  b_0(T_i-T)\right), \label{EQ:B3}\\
a_3(T,T_N)  &:= \sum_{i=1}^N  P(T,T_i)e^{-\alpha(T_i-T)} a_0(T_i-T), \label{EQ:A3}
\end{align}
with $b_0(.), a_0(.)$ given in Proposition \ref{Prop:SBZCPricing}.
\end{proposition}

Note that the joint survival annuity \eqref{EQ:JointSurvivalAnnuity} also has a linear-rational form so its density can be computed explicitly (up to a one-dimensional numerical integration).

\begin{proposition}\label{Prop:JointSurvivalAnnuityDensity} 
Let $\sum_{i=1}^N\mathrm{SB}(T,T_i)$ be a joint survival annuity defined in Proposition~\ref{Prop:JointSurvivalAnnuity}. Then the cumulative distribution function (conditional to time t) of $\sum_{i=1}^N\mathrm{SB}(T,T_i)$ is given by
\begin{align}
	\mathbb{P}\left(\sum_{i=1}^N\mathrm{SB}(T,T_i) \leq z \right) = \frac{1}{2}-\frac{1}{\pi}\int_0^{+\infty}\Im \left(\frac{e^{\mathrm{i}s(b_3 - z)} \Phi(T-t, \mathrm{i}s(a_3 - z u_0), v_t)}{s}\right) ds,\label{EQ:AnnuityDistribution}
\end{align}
where $a_3$ and $b_3$ are the expressions in Proposition~\ref{Prop:JointSurvivalAnnuity} with their dependencies dropped on $(T,T_N)$ and $\Phi(.,.,.)$ is given in Proposition \ref{Prop:MGFBru}. 

\end{proposition}

Since we have the density of the future annuity value, we can compute certain risk management quantities such as the lower tail that gives the value at risk or the expected shortfall. This sharply contrasts with the mortality model based on the exponential affine framework for which these quantities are very hard to obtain. Further along this line, an option on the joint survival annuity, also named the guaranteed joint survival annuity option, can be priced very efficiently in the linear-rational Wishart mortality model as the next section shows.

%%%%%%%%%%%%%%%%%%%%%%%%%%%%%%%%%%%%%%%%%%%%%%%%%%%%%%%%%%
%
% The joint survival annuity derivative and related quantities.
%
%%%%%%%%%%%%%%%%%%%%%%%%%%%%%%%%%%%%%%%%%%%%%%%%%%%%%%%%%%
\section{Guaranteed joint survival annuity option valuation}\label{PricingAnnuityDerivative}

With the joint survival annuity being priced, the next step is to calculate the value of an option on that contract and also to develop option price approximations that simplify the model implementation as well as the risk management of the product. First, we show that an option on the joint survival annuity is surprisingly simple to price in the linear-rational Wishart mortality model. 
\begin{proposition}\label{GAOpricing}
The value at time $t$ of a guaranteed joint survival annuity option \eqref{EQ:GAO2} with maturity $T$ is given by
\begin{align}
\bar{C}(t,T,T_N) := P(t,T)e^{-\alpha (T-t)}\frac{\mathbb{E}_t\left[(b_4(T,T_N) + \textup{tr}[a_4(T,T_N)v_T])_+\right] }{1+\textup{tr}[u_0v_0]},\label{EQ:GAOpricingFormula}
\end{align}
with $b_4(T,T_N)$ a constant and $a_4(T,T_N)$ a matrix such that
\begin{align}
b_4(T,T_N)  &:= b_3(T,T_N) - \frac{1}{g}, \label{EQ:B4}\\
a_4(T,T_N)  &:= a_3(T,T_N)   - \frac{1}{g}u_0, \label{EQ:A4}
\end{align}
where $b_3(.)$, $a_3(.)$ are given in Proposition \ref{Prop:JointSurvivalAnnuity}.
\end{proposition}

As formula \eqref{EQ:GAOpricingFormula} shows, the pricing of a guaranteed joint survival annuity option can be computed using an integration. Indeed, if $Y_T = b_4(T,T_N) + \textup{tr}[a_4(T,T_N)v_T] $ then its characteristic function is given by 
\begin{align}
\Phi_Y(z) &:=\mathbb{E}_t\left[e^{\mathrm{i} z Y_T} \right]= e^{\mathrm{i} z b_4(T,T_N)}\Phi(T-t,\mathrm{i} z a_4(T,T_N),v_0), \label{EQ:MGFz}
\end{align}
with $z\in \mathbb{C}$, $\mathrm{i}=\sqrt{-1}$ and the function $\Phi(.,.,.)$ given by \eqref{EQ:mgfBru}. The expectation in \eqref{EQ:GAOpricingFormula} can be computed as
\begin{align}  
\mathbb{E}_t\left[(Y_T)_+\right]=\frac{1}{\pi}\int_0^{+\infty} \Re\left( \frac{\Phi_Y(z+\mathrm{i}z_i)}{(\mathrm{i}(z+\mathrm{i}z_i))^2} \right) dz\,, \label{EQ:GAOFft}
\end{align}
with $z_i<0$ and $\Re(.)$ the real part.\\

As the result above shows, the valuation of the guaranteed joint survival annuity option only requires one-dimensional integration. Notice that the result holds whether $\omega =\beta \sigma^2$ (\textit{i.e.}, Bru case) is satisfied or not (see \cite{DaFonseca2024} for an implementation in the non-Bru case). In that latter case, computing the moment generating function involves further numerical integration. Even if the Bru case condition can be a bit too restrictive in practice, it offers nonetheless certain advantages in terms of option price approximations that provide quick alternatives to the exact formula. In what follows we derive three different approximations to the guaranteed joint survival annuity option price. These analytic approximations are new and would not be available in \cite{DaFonseca2024}.\\

To compute an approximation to the expectation \eqref{EQ:GAOpricingFormula}, one needs the moments of the variable $Y_T = b_4(T,T_N) + \textup{tr}[a_4(T,T_N)v_T]$, which by definition can be obtained by taking successive derivatives of the moment generating function. These derivatives can be tedious to perform in the general case, see for example \citet{LetacMassam2008}. Fortunately, for our Bru case, it is possible to obtain the moments by a direct series expansion.

\begin{proposition}\label{Prop:CumulantsExpansion}
Suppose that $\omega$ in  \eqref{EQ:Wishart} is such that $\omega=\beta \sigma^2$ with $\beta \in \mathbb{R}$ and $\beta \geq n+1$ and let $Y_T = b_4(T,T_N) + \textup{tr}[a_4(T,T_N)v_T] $ with $ b_4(T,T_N)$ and $ a_4(T,T_N)$ (written $ b_4$ and $ a_4$ for simplicity)  as in ~\eqref{EQ:GAOpricingFormula}. The cumulants of  $Y_T$ denoted $(\kappa_j)_{j \in \mathbb{N}_+}$ are given by
\begin{align}
\kappa_j 	&= b_4 \delta_{j1} +\beta (j-1)!2^{j-1}\textup{tr}[(\varsigma_Ta_4)^{j}]+ j! 2^{j-1} \textup{tr}[\vartheta_T^\top(\varsigma_Ta_4)^{j}],
\end{align}
with $\delta_{j1}=1$ if $j=1$ and $0$ otherwise, $\varsigma_T$ and $\vartheta_T$ given in Proposition \ref{Prop:MGFBru}. The moments of $Y_T$, denoted $(\mu_j)_{j\in \mathbb{N}_+}$, can be obtained from the cumulants thanks to the relation
\begin{align} 
\mu_{j}=\sum _{k=1}^{j}B_{j,k}(\kappa_{1},\ldots ,\kappa_{j-k+1}),\label{eq:MomentsFromCumulants}
\end{align}
with $B_{j,k}$ the incomplete Bell polynomials.
\end{proposition} 

The cumulants/moments enable the use of \citet{JarrowRudd1982}'s option price approximation that was also applied in \citet{CollinDufresneGoldstein2002} to the problem of swaption pricing, which is similar  to the pricing of an option on a coupon bearing bond or, equivalently, to an option on an annuity. We follow \citet{CollinDufresneGoldstein2002} in our presentation below. 
\begin{proposition}\label{GAOPricingApproximation}
Let $Y_T = b_4(T,T_N) + \textup{tr}[a_4(T,T_N)v_T] $ with $ b_4(T,T_N)$ and $ a_4(T,T_N)$ as in ~\eqref{EQ:GAOpricingFormula}, and denote the first three cumulants of $Y_T$, given Proposition \ref{Prop:CumulantsExpansion}, to be $(\kappa_1,\kappa_2,\kappa_3)$. Assume that the density of $Y_T$ can be approximated by
\begin{align}
  \frac{1}{\sqrt{2\pi \kappa_2}} e^{-\frac{(z-\kappa_1)^2}{2\kappa_2}}\left(\sum_{j=0}^3\eta_j (z-\kappa_1)^j\right)\,, \label{EQ:ApproximationDensity}
\end{align}
with $\lbrace \eta_j;j\in 0,\ldots, 3\rbrace$ some constants related to the first three cumulants of $Y_{T}$. Then
\begin{align}
\mathbb{E}_t\left[   \left(  Y_{T}\right)_{+}\right] \sim \sum_{j = 0}^3 \eta_j \xi_{j+1}+ \kappa_1\sum_{j = 0 }^{3} \eta_j \xi_{j} \label{EQ:ApproximationSwaption}\,,
\end{align}
where $\lbrace \xi_j;j= 0,\ldots, 4 \rbrace$  are some constants that can be computed explicitly and depend on the first three cumulants of $Y_{T}$.
\end{proposition}

The above proposition approximates the density of $Y_T$, the variable involved in the guaranteed joint survival annuity option, by a perturbation of the Gaussian distribution. Suppose that the eigenvalues of $a_4(T,T_N) + a_4(T,T_N)^\top$ are non-negative, $a_4(T,T_N) + a_4(T,T_N)^\top \in \mathbb{S}_n^{+}$ and hence $\mathrm{tr}[a_4(T,T_N)v_T]>0$. In this case, the support of the density of $Y_T$ is on the half line and using a Gaussian approximation might not be the most natural choice.\\

In the Bru case, the Wishart process has a marginal distribution that follows a non-central matrix chi-squared distribution or non-central Wishart distribution. Further can be said when that distribution is projected on a rank one matrix, it appears in \citet[Proposition 3.1]{LetacMassam2008} and reads as follows.
\begin{proposition}\label{Prop:MGFChiSquared}
Suppose that $\omega$ in  \eqref{EQ:Wishart} is such that $\omega=\beta \sigma^2$ with $\beta \in \mathbb{R}$ and $\beta \geq n+1$ and given a vector $\gamma\in \mathbb{R}^n$ such that $\|\gamma\|=1$ then the (scalar) variable $\gamma^\top v_T \gamma/(\varsigma_T)_{11}$ admits a non-central Chi-squared distribution with degrees of freedom $\beta$ and non-centrality parameter $(\vartheta_T^\top)_{11}$.
\end{proposition} 

The previous proposition gives us an approximation to the characteristic function of the Wishart process for a given $T$ when it is projected onto a rank-one matrix. In fact, in order to apply this proposition, one needs to specify a vector along which the state variable (\textit{i.e.}, the Wishart process) is projected. The following proposition utilises Proposition~\ref{Prop:MGFChiSquared} to obtain an approximation by projecting the state variable in the direction of the eigenvectors associated with the spectral decomposition of $\left(a_4(T,T_N) + a_4(T,T_N)^{\top}\right)/2$ and then summing the distributions obtained from the previous proposition under the assumption of independence. This assumption is inconsistent with the state variable distribution, but can still be used to obtain a numerical approximation as detailed below.

\begin{proposition}\label{Prop:ProjectionOnEig}
Let $Y_T = b_4(T,T_N)+\textup{tr}[a_4(T,T_N)v_T]$ of Proposition~\ref{GAOpricing} involved in the pricing of the guaranteed annuity option through~\eqref{EQ:MGFz} and~\eqref{EQ:GAOFft}. Denote the spectral decomposition of 
\begin{align}
\frac{a_4+a_4^\top }{2} = \sum_{i=1}^n \lambda_i \gamma_i\gamma_i^\top, \label{EQ:SpectralA4}
\end{align}  
with $(\gamma_i)_{i=1,\ldots, n}$ some orthonormal vectors and $(\lambda_i)_{i=1,\ldots, n}$ the corresponding (real) eigenvalues (we drop the dependency of $a_4$ and $b_4$ on $(T,T_N)$). Rewrite $Y_T$ as
\begin{align}
Y_T = b_4+\sum_{i=1}^n \lambda_i  \gamma_i^\top v_T \gamma_i,
\end{align}
then according to Proposition \ref{Prop:MGFChiSquared} for each $i=1,\ldots, n$ $\gamma_i^\top v_T \gamma_i/(\varsigma_T)_{ii}$ has a non-central Chi-squared distribution with degrees of freedom $\beta$ and non-centrality parameter $(\vartheta_T^\top)_{ii}$. Then  we approximate $Y_T$ with
\begin{align}
\tilde{Y}_T = b_4+\sum_{i=1}^n \lambda_i (\varsigma_T)_{ii} \chi_i^2, \label{eq:tildeY_T}
\end{align}
where $\left(\chi_i^2(\beta,(\vartheta_T^\top)_{ii})\right)_{i=1,\ldots, n}$ is a set of independent non-central Chi-squared random variables (where $\beta$ is the degrees of freedom and each variable has a specific non-centrality parameter $(\vartheta_T^\top)_{ii}$). In this case, $\mathbb{E}_t[(Y_T)_+]\sim \mathbb{E}_t[(\tilde{Y}_T)_+]$ and this latter expectation can be computed using \eqref{EQ:GAOFft} with the characteristic function of the variable $\tilde{Y}_T $ that depends (linearly) on a generalized Chi-squared distribution.
\end{proposition}

Note that if one eigenvalue in \eqref{EQ:SpectralA4} dominates all the others, then the above proposition can be simplified, since the computation of the guaranteed joint survival annuity option will then only require the integration of the non-central Chi-squared distribution, instead of requiring a Fourier transform.

\begin{corollary}\label{Prop:ProjectionOnOneEig}
Let $Y_T = b_4(T,T_N)+\textup{tr}[a_4(T,T_N)v_T]$ of Proposition~\ref{GAOpricing} involved in the pricing of the guaranteed annuity option through~\eqref{EQ:MGFz} and~\eqref{EQ:GAOFft}. Assume the spectral decomposition \eqref{EQ:SpectralA4} and suppose that we have $\lambda_1\geq \lambda_2 \geq \ldots \geq 0 \geq \ldots \geq \lambda_n$ and $\lambda_1\gg  |\lambda_j|,\, j = 2,\ldots n$ then approximate $Y_T$ with  
\begin{align}
\tilde{Y}_T = b_4+\lambda_1 \gamma_1^\top v_T \gamma_1\,, \label{EQ:tildeY}
\end{align}
and according to Proposition~\ref{Prop:MGFChiSquared} the scalar variable $\gamma_1^\top v_T \gamma_1/(\varsigma_T)_{11}$ follows a non-central Chi-squared distribution with degrees of freedom $\beta$ and non-centrality parameter $(\vartheta_T^\top)_{11}$.
%\begin{align}
%\mathbb{E}[(Y_T)_+]\sim \int_{z_*}^{+\infty}(b_4 +\lambda_1 (\varsigma_T)_{11} z)\frac{1}{2}e^{-\frac{1}{2}(\sqrt{z}-\sqrt{(\vartheta_T^\top)_{11}})^2}\left(\frac{z}{(\vartheta_T^\top)_{11}}\right)^{\frac{\beta}{4} -\frac{1}{2}} \bar{I}_{\beta/2 -1}\left(\sqrt{z(\vartheta_T^\top)_{11}}\right) dz\,.
%\end{align}
%where $z_* =-b_4/((\varsigma_T)_{11}\lambda_1)$.\\

If $\lambda_1\geq \lambda_2 \geq \ldots \geq 0 \geq \ldots \geq \lambda_n$ and $|\lambda_n| \gg |\lambda_j|,\, j = 1,\ldots n-1$ then approximate $Y_T$ with
\begin{align}
\tilde{Y}_T = b_4 + \lambda_n \gamma_n^\top v_T \gamma_n\,, \label{EQ:tildeY2}
\end{align}
and according to Proposition~\ref{Prop:MGFChiSquared} the scalar variable  $\gamma_n^\top v_T \gamma_n/(\varsigma_T)_{11}$ follows a non-central Chi-squared distribution distribution with degrees of freedom $\beta$ and non-centrality parameter $(\vartheta_T^\top)_{11}$.
%
%then 
%\begin{align}
%\mathbb{E}[(Y_T)_+]\sim \int_0^{z^*}(b_4 +\lambda_n (\varsigma_T)_{11} z)\frac{1}{2}e^{-\frac{1}{2}(\sqrt{z}-\sqrt{(\vartheta_T^\top)_{11}})^2}\left(\frac{z}{(\vartheta_T^\top)_{11}}\right)^{\frac{\beta}{4} -\frac{1}{2}} \bar{I}_{\beta/2 -1}\left(\sqrt{z(\vartheta_T^\top)_{11}}\right) dz\,,
%\end{align}
%where $z^* =-b_4/((\varsigma_T)_{11}\lambda_n)$ and $\bar{I}_{k}\left(z\right)$ the exponentially scaled modified Bessel function of the first kind.\footnote{The use of the  exponentially scaled modified Bessel function of the first kind is recommended as it leads to a more robust numerical integration.} 
\end{corollary}

%%%%%%%%%%%%%%%%%%%%%%%%%%%%%%%%%%%%%%%%%%%%%%%%%
%
% Gamma approximation
%
%%%%%%%%%%%%%%%%%%%%%%%%%%%%%%%%%%%%%%%%%%%%%%%%%
In the equation \eqref{EQ:ApproximationDensity} we approximate the density of $Y_T$ using a perturbation of the Gaussian distribution, but the marginal distribution of the Wishart process suggests that it is more natural to approximate the density with a perturbation of the (scalar) gamma distribution. Indeed, if $a_4(T,T_N)$ in Proposition \ref{Prop:CumulantsExpansion} has positive eigenvalues, then necessarily $b_4<0$ and the support of the density of $Y_T$ is in $[k,\; +\infty)$ with $k<0$ and if $a_4(T,T_N)$  has negative eigenvalues, then necessarily $b_4>0$ and the support of the density of $Y_T$ is in $(-\infty,\; k]$ for $k>0$. As such, the Gaussian distribution, whose support is in $\mathbb{R}$, might not be the most convenient distribution to approximate the density of $Y_T$. Further to this, taking the limit as \(t \to \infty\) in \eqref{EQ:mgfBru} shows that the asymptotic distribution of the Wishart process has a moment generating function of the form \(\frac{1}{\det(I_n - 2\varsigma_t\theta_1)^{\beta/2}}\), which corresponds to a matrix gamma distribution. \citet{FilipovicMayerhoferSchneider2013} show how to approximate the density of a random variable using a perturbation of the (scalar) gamma distribution, instead of the Gaussian distribution in \eqref{EQ:ApproximationDensity}, if the cumulants of the variable are available. As such, Proposition \ref{Prop:CumulantsExpansion} allows us to use their results as the following proposition shows. 
\begin{proposition}\label{GAOPricingApproximationGamma}
Let $Y_T = b_4(T,T_N) + \textup{tr}[a_4(T,T_N)v_T] $ with $ b_4(T,T_N)$ and $ a_4(T,T_N)$ as in ~\eqref{EQ:GAOpricingFormula}, and denote the first three cumulants of $Y_T$, given Proposition \ref{Prop:CumulantsExpansion}, to be $(\kappa_1,\kappa_2,\kappa_3)$. If the eigenvalues of $ a_4(T,T_N)$ are all positive, then we can obtain the following approximation 
\begin{align}
\mathbb{E}_t\left[   \left(  Y_{T}\right)_{+}\right]&= \sum_{j=0}^{3}\sum_{i=0}^{j}  \frac{c_j\gamma_{j,i}}{\bar{\beta}} ( \bar{\alpha}+1)_{i+1} Q(\bar{\alpha}+i+2,-\bar{\beta} k)  \nonumber \\
&+ \sum_{j=0}^{3}\sum_{i=0}^{j}  c_j\gamma_{j,i}k ( \bar{\alpha}+1)_{i} Q(\bar{\alpha}+i+1,-\bar{\beta} k)\,, \label{EQ:GammaApproximationJumpPositive}
\end{align}
with $\lbrace c_j; j =0,\ldots, 3 \rbrace$, $\lbrace  \gamma_{j,i}; j,i =0,\ldots, 3 \rbrace$, $\bar{\alpha}$ and $\bar{\beta}$ some parameters that depend on $k$ and the first three moments of $Y_{T}$, while $Q(s,x)$ stands for the upper regularised incomplete gamma function while $(x)_n$ is the Pochhammer function.\\

If the eigenvalues of $ a_4(T,T_N)$ are all negative, then we can obtain the following approximation 
\begin{align}
\mathbb{E}_t\left[   \left(  Y_{T}\right)_{+}\right]&= -\sum_{j=0}^{3}\sum_{i=0}^{j}  \frac{c_j\gamma_{j,i}}{\bar{\beta}} ( \bar{\alpha}+1)_{i+1} P(\bar{\alpha}+i+2,-\bar{\beta} k)  \nonumber\\
&- \sum_{j=0}^{3}\sum_{i=0}^{j}  c_j\gamma_{j,i}k ( \bar{\alpha}+1)_{i} P(\bar{\alpha}+i+1,-\bar{\beta} k)\,, \label{EQ:GammaApproximationJumpPositive}
\end{align}
with $\lbrace c_j; j =0,\ldots, 3 \rbrace$, $\lbrace  \gamma_{j,i}; j,i =0,\ldots, 3 \rbrace$, $\bar{\alpha}$ and $\bar{\beta}$ some parameters that depend on $k$ and the first three moments of $Y_{T}$, while $P(s,x)$ stands for the lower regularised incomplete gamma function while $(x)_n$ is the Pochhammer function.
\end{proposition}

\section{Model implementation}\label{Section:num_experiments}

In this section we present numerical experiments on the linear-rational Wishart mortality model. Our numerical experiments consist of sensitivity analyses of the guaranteed joint survival annuity option price on selected parameters, it also compares the three approximation methods developed in Section~\ref{PricingAnnuityDerivative} to the characteristic function pricing approach, and we also demonstrate the potential consequences of incorrectly not accounting for the dependency between annuitants demonstrated in the guaranteed joint survival annuity option price. For simpler exposition, we will refer to the guaranteed joint survival annuity option price as the option price. \\

The parameters utilised in our numerical experiments are presented in Tables~\ref{tab:model_param_values} and \ref{tab:GSA_param_values}. Table~\ref{tab:model_param_values} encapsulates model-specific parameters. These parameters were judiciously chosen to be consistent with values empirically observed in the extant literature.\myfootnote{We take $r=0$ to analyse the effect of the Wishart parameters on the option price without contaminating influence from the interest rate discounting factor.} Meanwhile Table~\ref{tab:GSA_param_values} contains the foundational values on which our subsequent sensitivity analyses are predicated.\\

\begin{center}
	[ Insert Table~\ref{tab:model_param_values} here ]
\end{center}

\begin{center}
	[ Insert Table~\ref{tab:GSA_param_values} here ]
\end{center}

To elucidate our model parameter selection, let us consider Table~\ref{tab:moments_mux_muy}. This table delineates the moments for $\mu_x(s)$ and $\mu_y(s)$ at various time points $s$. These moments were derived using Corollary~\ref{Prop:MortalityIntensityDensity}. Specifically, when $s=2$, we observe that the anticipated mortality rates for annuitants $x$ and $y$ stand at \( 0.0199 \) and \( 0.0198 \) respectively. In a recent study by \citet{LiLiuTangYuan2023}, the authors utilised U.S. mortality data spanning 2017-2022 and find that the excess mortality rate is approximately 0.0132.\myfootnote{This computation is based on the one-year conditional expectation based on parameters reported in Table 4.1 of \citet{LiLiuTangYuan2023} assuming an initial excess mortality of 0.} Augmenting our analysis, the one-year conditional expected mortality rate in \citet{XuSherrisZiveyi2020b} is found to be 0.0107 and is based on data from the Human Mortality Database.\myfootnote{This calculation stems from the one-year conditional expectation of the one-factor mortality model reported in Table 1 of \citet{XuSherrisZiveyi2020b}.}\\

\begin{center}
	[ Insert Table~\ref{tab:moments_mux_muy} here ]
\end{center}

One compelling advantage of the linear-rational Wishart model is its capability to accommodate explicit dependence between annuitants. This essentially means that correlations in the longevity risks of two individuals can be accurately captured by the model.\myfootnote{An important note to make here is that our model can be seamlessly adapted to address larger groups, be it cohorts, entire populations or even distinct risk factors.} Drawing upon \eqref{EQ:v11v22} we can compute the instantaneous correlation between the normalised intensities $\mu_x^*(s) = (1+\textup{tr}[u_0v_s])\mu_x(s)$ and $\mu_y^*(s)= (1+\textup{tr}[u_0v_s])\mu_y(s)$. For clarity, this computation is explicated below and is derived using Proposition~\ref{Prop:QuadraticVariation} and \eqref{EQ:intensity_covariation}:

\begin{align}
	\text{Corr}(d\mu_x^*(s), d\mu_y^*(s)) &= \frac{\Cov(d\mu_x^*(s), d\mu_y^*(s))}{\sqrt{\Cov(d\mu_x^*(s), d\mu_x^*(s)) \Cov(d\mu_y^*(s), d\mu_y^*(s))}}\nonumber\\
	&= \frac{v_{12,s}(\sigma^2)_{12}}{\sqrt{v_{11,s}v_{22,s}(\sigma^2)_{11}(\sigma^2)_{22}}}ds.\label{EQ:instant_correlation}
\end{align}

Armed with the values provided in Table~\ref{tab:model_param_values}, we can discern an instantaneous correlation between the normalised intensities of 0.40 at the inception \( t\). Due to the choice of our matrix $m$, we can also investigate the average `asymptotic' instantaneous correlation between the two annuitants. By substituting the formula \( \lim_{s\rightarrow \infty} \mathbb{E}[v_{s}] \) from Lemma \ref{LEM:TraceMean} into \eqref{EQ:instant_correlation}, our model posits an average asymptotic correlation between the normalised intensities of 0.65. The values of our correlations are reasonable and well-founded. The work \citet{FreesCarriereValdez1996}, which pegged the correlation between annuitants in last-survivor annuities within a 95\% confidence band of (0.41, 0.56) based on last survivor annuity contracts provided by a Canadian insurer.\\

Lastly, we focus on the parameter \( g \) outlined in Table~\ref{tab:GSA_param_values}. It was carefully selected to make the guaranteed joint survival annuity option nearly `at-the-money' for values of \( T-t=2 \) and \( T_N -t= 5 \), considering the model parameters from Table~\ref{tab:model_param_values}. To determine this parameter, we calculated the expected value of a joint survival annuity spanning five yearly payments, denoted as:
\[
\sum_{i=1}^5 \mathrm{SB}(t, T_i),
\]
where \( T_i-t = i \) for \(i= 1,...,5 \). Using Proposition~\ref{Prop:JointSurvivalAnnuity}, the computed value of the joint survival annuity is 4.467. Consequently, we chose \( g = 1/4.467 \approx 0.225 \), aligning the value of the joint survival annuity starting in 2 years with that of a joint survival annuity beginning immediately.

\subsection{Sensitivity analyses}
In this section, we study the sensitivity of the guaranteed joint survival annuity option price in relation to changes in various model parameters. Figure~\ref{Fig:Sensitivity_graphs} displays the variations in option price as parameters $\alpha$, $\beta$, $\rho_{\sigma}$, $m_{11}$, $\sigma_{11}$, and $T$ are adjusted.\\
\begin{center}
 [ Insert Figure~\ref{Fig:Sensitivity_graphs} here ]
\end{center}
The option price displays a nearly linear increase when parameters $\beta$, $\rho_{\sigma}$, $m_{11}$, and $\sigma_{11}$ are changed. This is consistent with expectations and is attributable to the interplay of mortality intensities as described in \eqref{EQ:mu_x_intensity_simple}-\eqref{EQ:mu_y_intensity_simple}. For instance, an increase in $\beta$ directly increases $\omega$ since our model follows \eqref{EQ:BruOmegaSpecification}, leading to decreased instantaneous mortality intensities. This results in a greater likelihood of option exercise due to increased joint survival probability, thus raising the option price. Similar reasoning applies to the parameters $m_{11}$ and $\sigma_{11}$. Given the relationship $\omega = \beta \sigma^2$, we derive:
\begin{align*}
	\omega_{11} &= \beta \left((\sigma_{11})^2 + (\sigma_{12})^2 \right),\\
	\omega_{22} &= \beta \left((\sigma_{22})^2 + (\sigma_{12})^2 \right).
\end{align*}
Since $\sigma_{12} = \rho_{\sigma}\sqrt{\sigma_{11} \sigma_{22}}$, increasing $\rho_{\sigma}$ reduces the mortality intensity by increasing $\sigma_{12}$, which subsequently increases $\omega_{11}$ and $\omega_{22}$ and hence leads to an increase in the option price.\\

In stark contrast, the relationship between the option price and $\alpha$ is notably non-linear and decreasing. A mere 10\% drop in $\alpha$ from 0.04 to 0.036 results in an almost five-fold surge in the option price, underscoring $\alpha$ as a pivotal parameter in the linear-rational Wishart model. This parameter is from the potential approach of \citet{Rogers1997} which dictates to specify the state-price density first, and then get mortality dynamics that are consistent with the specified state-price density, contrasted to the standard method of specifying mortality intensities first and then derive the state-price density. Since the state-price density is specified first, we are able to enforce positive mortality intensities, a significant advantage of our model and approach. Previous models either struggled with ensuring positive mortality intensities or lacked a versatile correlation structure. Although at first glance, changes in the option price due to $\alpha$ should be linear, as seen in \eqref{EQ:mu_x_intensity_simple}-\eqref{EQ:mu_y_intensity_simple}, it is essential to realise that $\alpha$ influences the option price at multiple points, as evident from \eqref{EQ:GAOpricingFormula}-\eqref{EQ:A4}, amplifying its impact.\\

Lastly, the relationship with $T$ is non-linear and decreasing. This decrease is intuitive: annuitants have a higher mortality risk over extended periods. The non-linearity stems from the fact that the option's expiry date, $T$, affects the pricing equation at various stages, leading to a cumulative effect on the option price.

\subsection{Approximation methods}
Figure~\ref{Fig:Approximation_methods} showcases the effectiveness of the three approximation methods introduced in Proposition~\ref{GAOPricingApproximation}, Proposition~\ref{Prop:ProjectionOnEig} and Proposition~\ref{GAOPricingApproximationGamma} benchmarked against the characteristic function approach across various strike values. Notably, the gamma approximation from Proposition~\ref{GAOPricingApproximationGamma} outperforms the Gaussian approximation of Proposition~\ref{GAOPricingApproximation} and the spectral decomposition approximation outlined in Proposition~\ref{Prop:ProjectionOnEig}. The right panel of Figure~\ref{Fig:Approximation_methods} reveals that the absolute percentage error for the gamma method is consistently lower than the Gaussian and the spectral decomposition approximation errors. Logically, this makes sense. The asymptotic distribution of the Wishart process is a matrix gamma distribution, and thus for the scalar gamma approximation to do well is of no surprise. Interestingly, while the Gaussian approximation tends to over-estimate the option price, the spectral decomposition approximation tends to under-estimate it, while the gamma approximation initially over-estimates at low strikes and under-estimates at higher strikes, albeit very slightly. However, the Gaussian and spectral approximation methods display enhanced accuracy as the guaranteed survival annuity rate $g$ rises, \textit{i.e.}, when the option is moving in-the-money.\\
	
	Even though the gamma approximation method outperforms the Gaussian and spectral methods for these numerical examples, in other parameters the other methods may outperform the gamma method. For example, if one eigenvalue of $a_4$ from \eqref{EQ:A4} dominates its other eigenvalues then it is possible for the spectral method to perform the best. Similarly, if we are examining a problem where some of the eigenvalues are positive and some are negative, then the Gaussian approximation may perform the best.
	 Nevertheless, the predominant benefit of these approximation methods is their rapid computational capability, with the Gaussian and gamma approximations being especially efficient. These approximations only necessitate computations of constants and the use of the Gaussian cumulative distribution and density functions or the gamma functions, which are inherently swift.

\begin{center}
	[ Insert Figure~\ref{Fig:Approximation_methods} here ]
\end{center}

\subsection{Independence vs dependence}\label{Section:Independence vs dependence}
To underscore the model risk associated with assumptions about dependencies between annuitants, we examine the differences between the outcomes when assuming either independence or dependence. Figure~\ref{Fig:Independence_depdendence} showcases the option price under two scenarios:
\begin{enumerate}
	\item A positive correlation between the two annuitants' mortality rates, reflected by the model parameters in Table~\ref{tab:model_param_values} (shown as blue solid).
	\item An assumption of independence between the two annuitants' mortalities (shown as red dashed).
\end{enumerate}

In the second scenario, independence is represented by setting \(\rho_{\sigma}=0\). This results in the quadratic covariation between the normalised mortality intensities \(d\langle \mu_x^*(\cdot), \mu_y^*(\cdot) \rangle_s = 0\), as seen in \eqref{EQ:instant_correlation}.\myfootnote{While \(d\langle \mu_x^*(\cdot), \mu_y^*(\cdot) \rangle_s = 0\), it is crucial to understand that the mortality intensities of the annuitants will not be independent. As depicted in \eqref{EQ:mu_x_intensity_simple}-\eqref{EQ:mu_y_intensity_simple}, both individual mortalities rely on the common denominator $1 + \tr[u_0v_s]$. Thus, they cannot be mutually independent through the setting of \(\rho_{\sigma}=0\) alone. By adjusting $\rho_{\sigma}$, we can approximate independence by nullifying the quadratic covariation between \(\mu_x^*(s)\) and \(\mu_y^*(s)\).} But, setting \(\sigma_{12}\) to zero indirectly alters the quadratic variations of \(v_{11,s}\) and \(v_{22,s}\) as seen in \eqref{EQ:v11v22}. To account for this, we introduce a new process \(\tilde{v}\) built upon the parameters in Table~\ref{tab:model_param_values}, from which we compute our independent option price on, except for the matrix \(\sigma\), which is modified to \(\sigma^{ind}\), defined by:
\begin{align}
	\sigma_{11}^{ind} & = \sqrt{\left(\sigma_{11}\right)^2 + \left(\sigma_{12}\right)^2},\label{EQ:sig_11_ind}\\
	\sigma_{22}^{ind} & = \sqrt{\left(\sigma_{22}\right)^2 + \left(\sigma_{12}\right)^2},\label{EQ:sig_22_ind}\\
	\sigma_{12}^{ind} & = 0.\label{EQ:sig_12_ind}
\end{align}
With these adjustments, $d\langle \tilde{v}_{11,.}, \tilde{v}_{11,.} \rangle_s /\tilde{v}_{11,s} $ aligns with $d\langle v_{11,.}, v_{11,.} \rangle_s / v_{11,s} $, and $d\langle \tilde{v}_{22,.}, \tilde{v}_{22,.} \rangle_s / \tilde{v}_{22,s}$ aligns with $d\langle v_{22,.}, v_{22,.} \rangle_s / v_{22,s}$. The key distinction is that $d\langle \tilde{v}_{11,.}, \tilde{v}_{22,.} \rangle_s = 0$, unlike the positive, non-zero value of $d\langle v_{11,.}, v_{22,.} \rangle_s$.\\

From Figure~\ref{Fig:Independence_depdendence}, it is evident that the option price, when assuming independence, is consistently lower than when accounting for positive dependence between the annuitants' mortality rates. This is logically consistent: the probability of both annuitants being alive is higher under positive dependence than under independence. Furthermore, the discrepancy in prices can be substantial, with the absolute percentage error reaching up to 40\% when the guaranteed survival annuity rate \(g\) is minimal. This underscores the significance of precisely modelling dependency, given its potential economic implications. This result extends previous findings showing the importance of mortality dependence on longevity derivatives, see \citet{FreesCarriereValdez1996} for joint / last survivor survival annuities, \citet{WangYangHuang2015} for longevity bonds or \citet{YangWang2013} for the survivor swaps.\myfootnote{Needless to say that any work looking at reducing longevity risks through (cross) hedging relies on this dependence, see for example \citet{Zhou2019}.} 

\begin{center}
	[ Insert Figure~\ref{Fig:Independence_depdendence} here ]
\end{center}

\section{Conclusion}\label{Conclusion}

In this paper, we introduce the linear-rational Wishart mortality model, marking a significant departure from traditional mortality models. One of the prominent features of our model is the integration of the Wishart process coupled with the potential approach delineated by \citet{Rogers1997}. This combination provides a versatile framework, enabling a comprehensive dependency between mortality intensities. Notably, these intensities remain positive inherently, offering a clear advantage over prior models that either struggled to ensure this positivity or were lacking in a general  correlation structure.\smallpar

A pivotal contribution lies in the derivation of closed-form solutions for both the joint survival annuity and the guaranteed joint survival annuity option, made possible through leveraging the strong analytical properties of the Wishart process and the linear-rational structure of the model. The density of the joint survival annuity is known, so that risk indicators can be easily computed. For the guaranteed joint survival annuity option price, its computation only requires a one-dimensional integration. Further building on the analytical properties of the Wishart process, we introduce three unique approximation techniques tailored for pricing the guaranteed joint survival annuity option. The first employs the cumulants of the Wishart process and approximates its density by a perturbation of the Gaussian distribution, the second is rooted in the spectral decomposition and the third approximation is based on a perturbation of the gamma distribution. All three approximations work well, with them providing reasonably accurate results and their swift numerical implementations make them particularly valuable for practical applications in pricing of the guaranteed joint survival annuity option.\smallpar

Looking ahead, there are several avenues for extension and refinement. The first is certainly the estimation of the model, where the density as well as the moments of the joint survival annuity will certainly be important. Note that the model parameters are matrices with some algebraic properties that require adapted optimisation procedures. Moreover, given the analytical attributes of our model, deriving option Greeks presents as a logical next step that will complement the risk management indicators presented in this work, further strengthening the practical applicability of our model.\smallpar

The inherent flexibility of our linear-rational Wishart mortality model extends beyond the confines of our current investigation. Its generic framework is adaptable to multi-population, multi-cohort, and even cross-asset type challenges. The derivations and results showcased in this paper could be expanded to embrace these diverse frameworks.

\clearpage

\section*{Declarations}

\begin{itemize}
\item The authors contributed equally to this work.
\item The authors report there are no competing interests to declare.
\item Data availability: Not applicable.
\item Funding organizations: Not applicable.
\end{itemize}

\clearpage
%\footnotesize
\linespread{1}
\selectfont
\bibliographystyle{abbrvnat}
\bibliography{Biblio}
\clearpage
\section{Figures and Tables}

\begin{table}[htbp]
	\caption{Wishart model parameter values.}\label{tab:model_param_values}
	\begin{center}
		\begin{tabular}{|c c|}
			\hline
			Variable & Value \\
			\hline
			$\alpha$ & 0.04 \\
			$r$ & 0\\
			$\beta$ & 3.5 \\
			$\rho_{\sigma}$ & 0.5 \\
			$\sigma_{12}$ & $\rho_{\sigma}  \sqrt{0.06 \times 0.04}$ \\
			$\sigma$ & $\begin{pmatrix} 0.06 & \sigma_{12} \\ \sigma_{12} & 0.04 \end{pmatrix}$ \\
			$m$ & $\begin{pmatrix} -1 & 0 \\ 0 & -1 \end{pmatrix}$ \\
			$v_{12}$ & 0.5 $\sqrt{0.005 \times 0.0025}$ \\
			$v_0$ & $\begin{pmatrix} 0.005 & v_{12} \\ v_{12} & 0.0025 \end{pmatrix}$ \\
			$u_1$ & $\begin{pmatrix} 1 & 0 \\ 0 & 0 \end{pmatrix}$ \\
			$u_2$ & $\begin{pmatrix} 0 & 0 \\ 0 & 1 \end{pmatrix}$ \\
			\hline
		\end{tabular}
	\end{center}
	{\footnotesize \textit{Note.} This table contains the model parameters used for our numerical experiments in Section~\ref{Section:num_experiments}.}
	
\end{table}

\begin{table}[htbp]
	\caption{Guaranteed joint survival annuity option parameters.}\label{tab:GSA_param_values}
	\begin{center}
		\begin{tabular}{|cc|}
			\hline
			Variable & Value \\
			\hline
			$T-t$ & 2\\
			$T_N-t$ & 5\\
			$g$ & 0.225\\
			$z$ & $-0.025$\\
			\hline
		\end{tabular}
	\end{center}
	{\footnotesize \textit{Note.} This table contains the parameters that our analyses on the guaranteed joint survival annuity option prices are centered on. The option is exercisable in $T-t=2$ years, with a joint survival annuity payable yearly for $T_N-t = 5$ years and guaranteed joint survival annuity rate $g=0.225$ which is chosen so the option is approximately at-the-money based on the parameters in Table~\ref{tab:model_param_values}.}
\end{table}

\begin{table}[htbp]
	\caption{Moments of $\mu_x(T)$ and $\mu_y(T)$.}\label{tab:moments_mux_muy}
	\begin{center}
		\begin{tabular}{|c c c c c|}
			\hline
			$T-t$ & 1 & 2 & 5 & 10 \\
			\hline
			$\mathbb{E}_t[\mu_x(T)]$ & 0.0193 & 0.0199 & 0.0200 & 0.0200 \\
			$\Var(\mu_x(T))$ & $1.06 \times 10^{-4}$ & $1.20 \times 10^{-4}$ & $1.20 \times 10^{-4}$ & $1.20 \times 10^{-4}$ \\
			$\mathbb{E}_t[\mu_y(T)]$ & 0.0195 & 0.0198 & 0.0198 & 0.0198 \\
			$\Var(\mu_y(T))$ & $2.59 \times 10^{-5}$ & $2.83 \times 10^{-5}$ & $2.86 \times 10^{-5}$ & $2.86 \times 10^{-5}$\\
			\hline		
		\end{tabular}
	\end{center}
	\footnotesize{\textit{Note.} This table contains the expectation and variance of $\mu_x(T)$ and $\mu_y(T)$ computed at the values of $T-t=1,2,5,10$. The moments were numerically calculated using \eqref{EQ:IntensityDensity}}.
\end{table}

\begin{figure}[htbp]
	\caption{Sensitivity analyses.}\label{Fig:Sensitivity_graphs}
	\begin{center}
		\includegraphics[width=7.5cm]{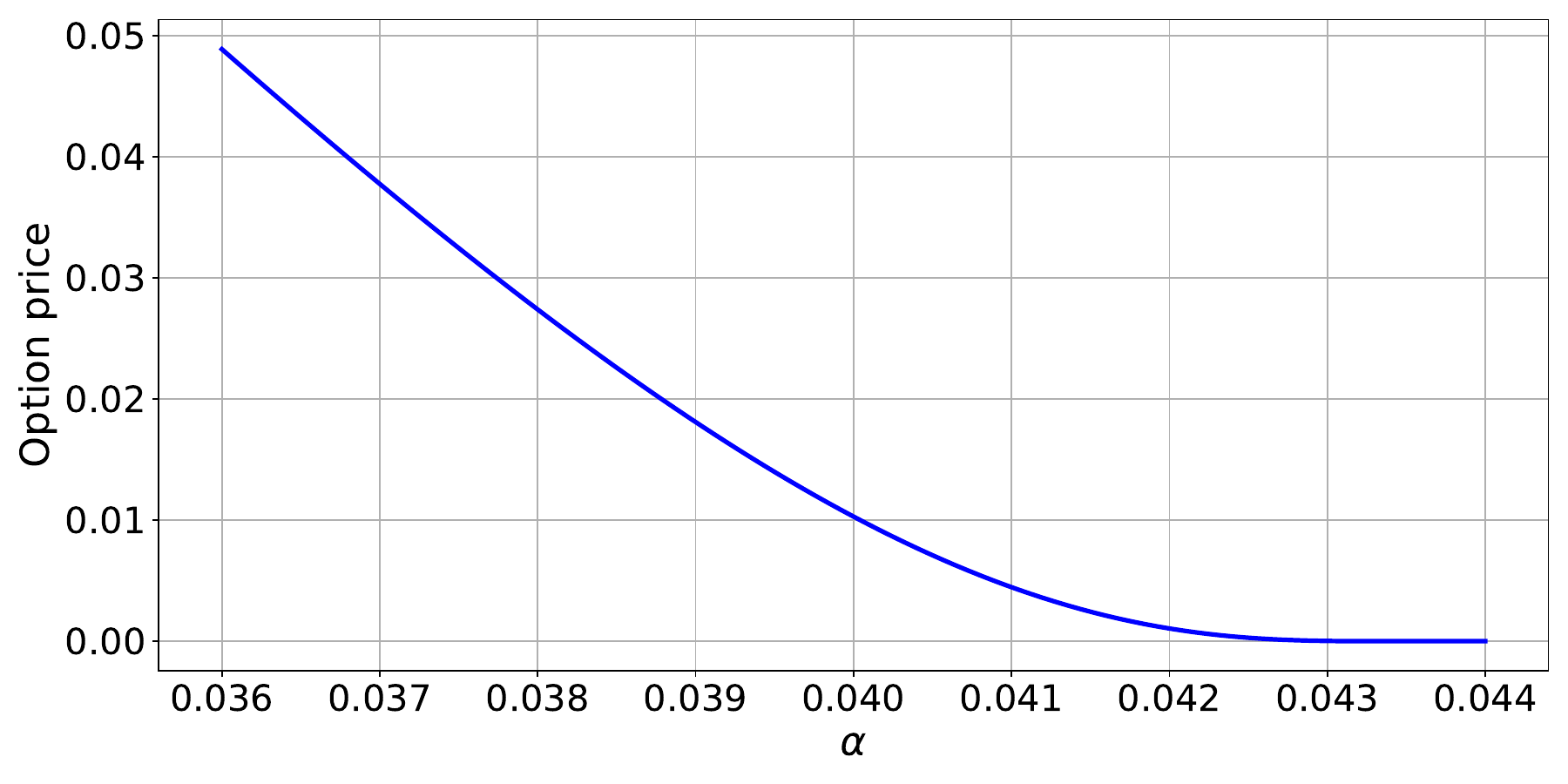}\includegraphics[width=7.5cm]{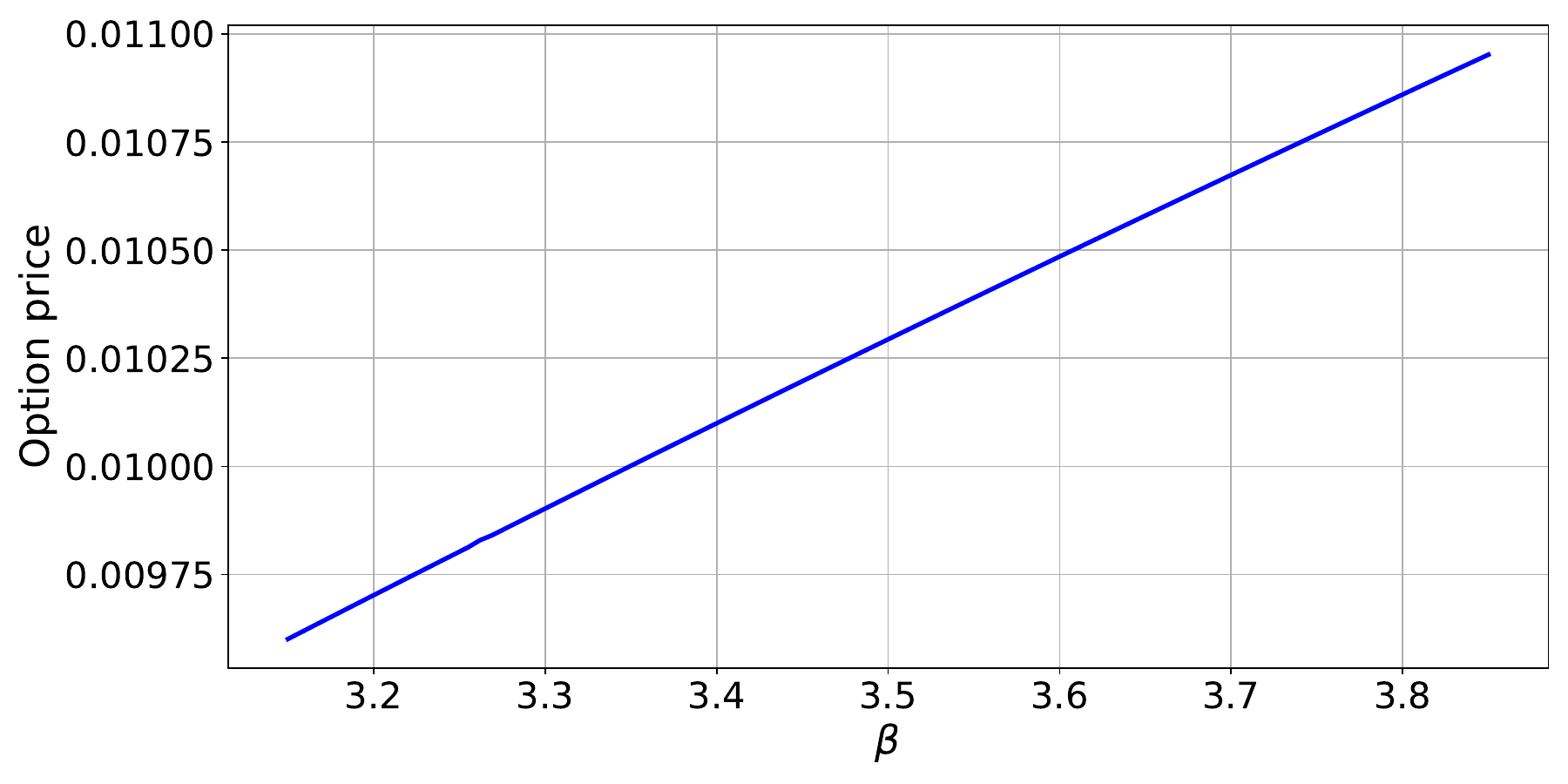}
		\includegraphics[width=7.5cm]{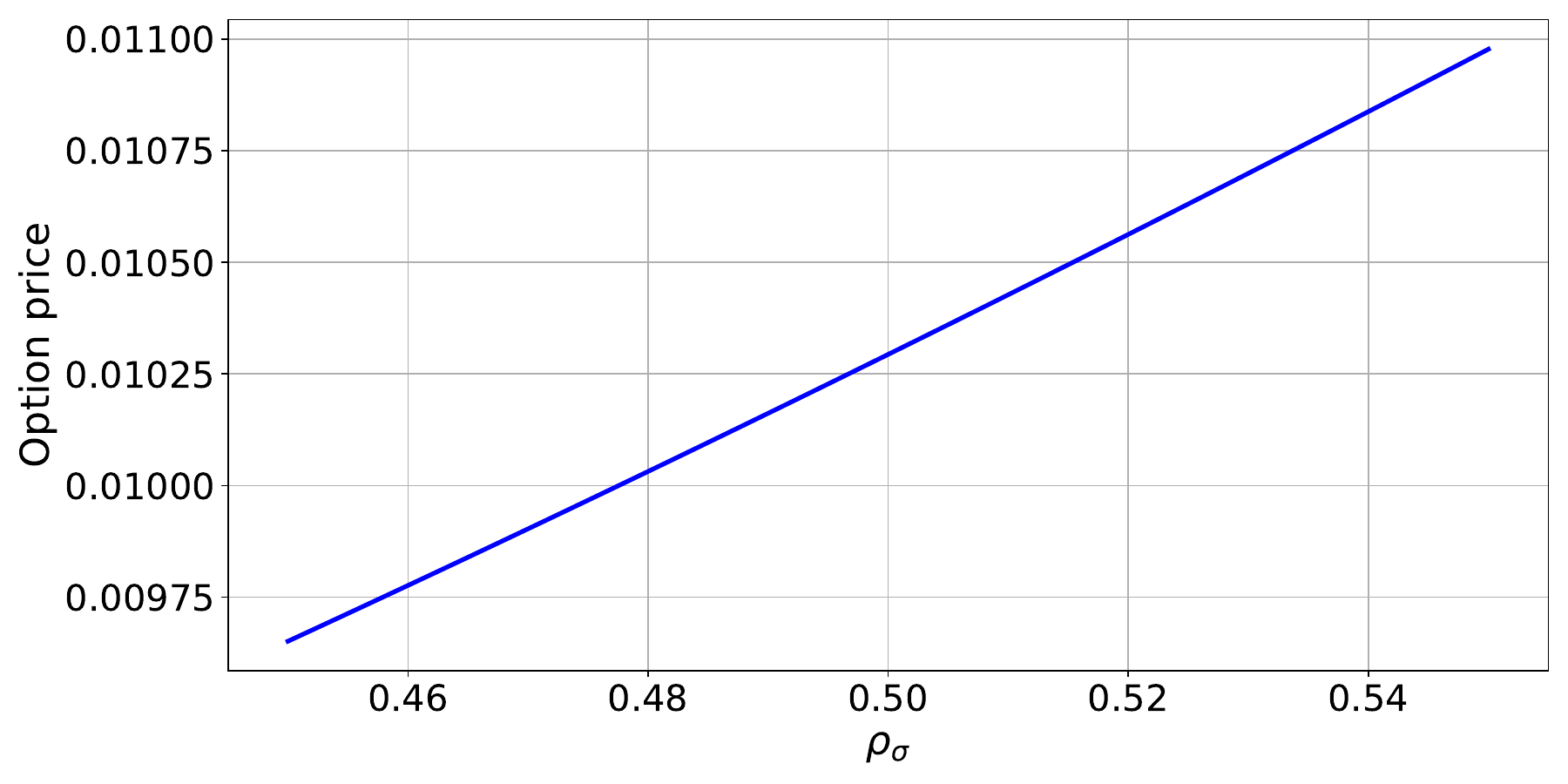}\includegraphics[width=7.5cm]{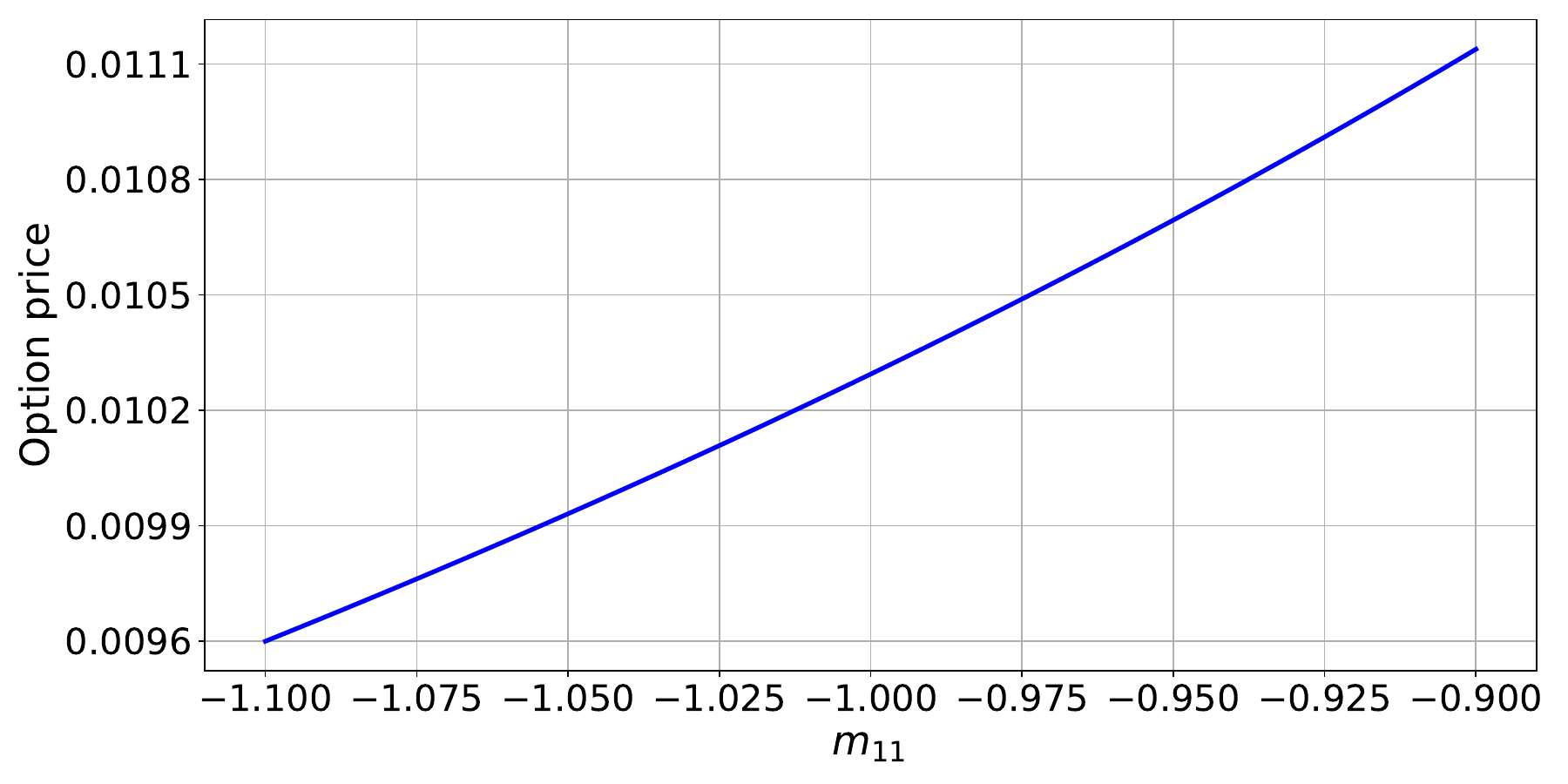}
		\includegraphics[width=7.5cm]{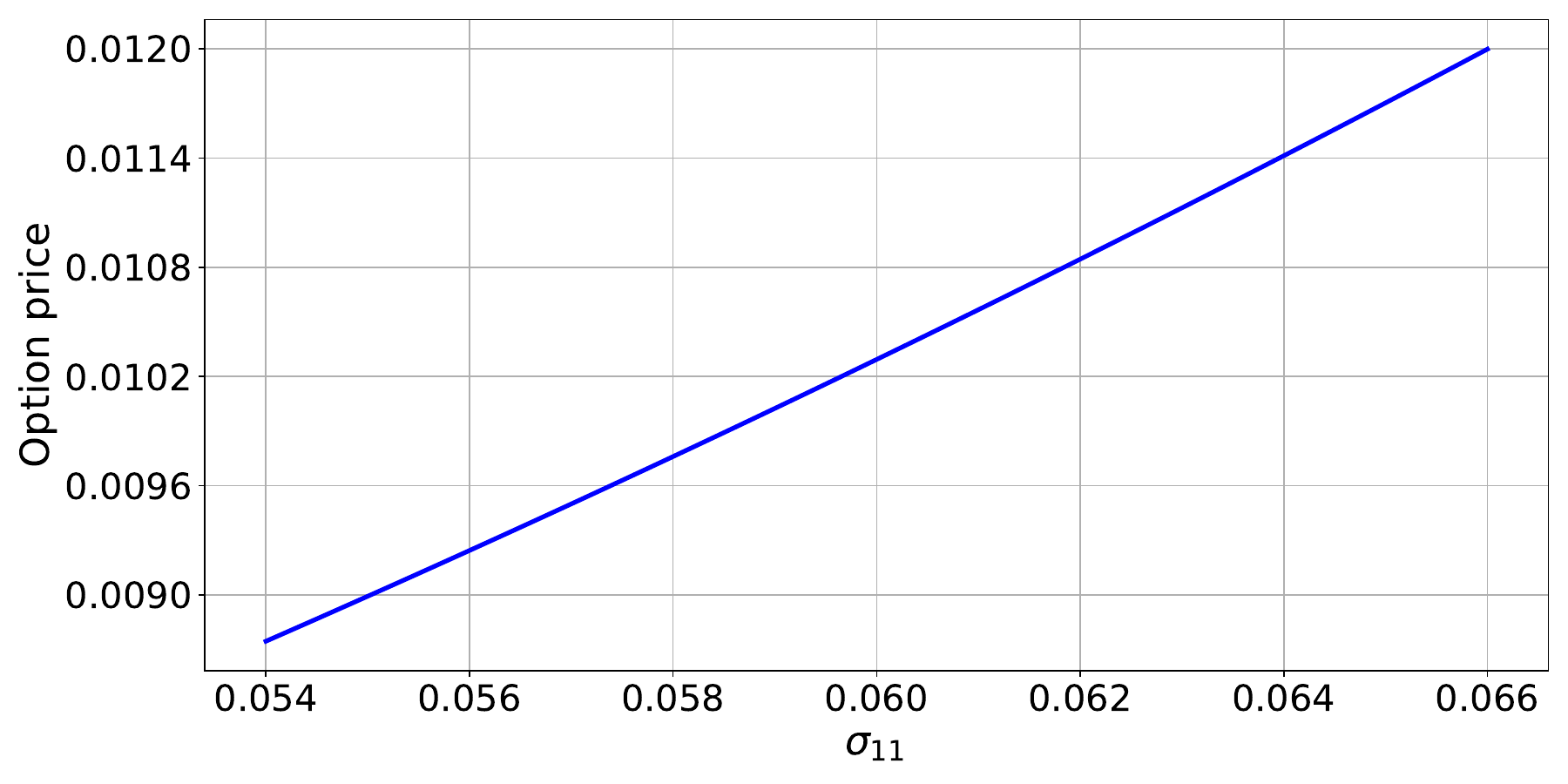}\includegraphics[width=7.5cm]{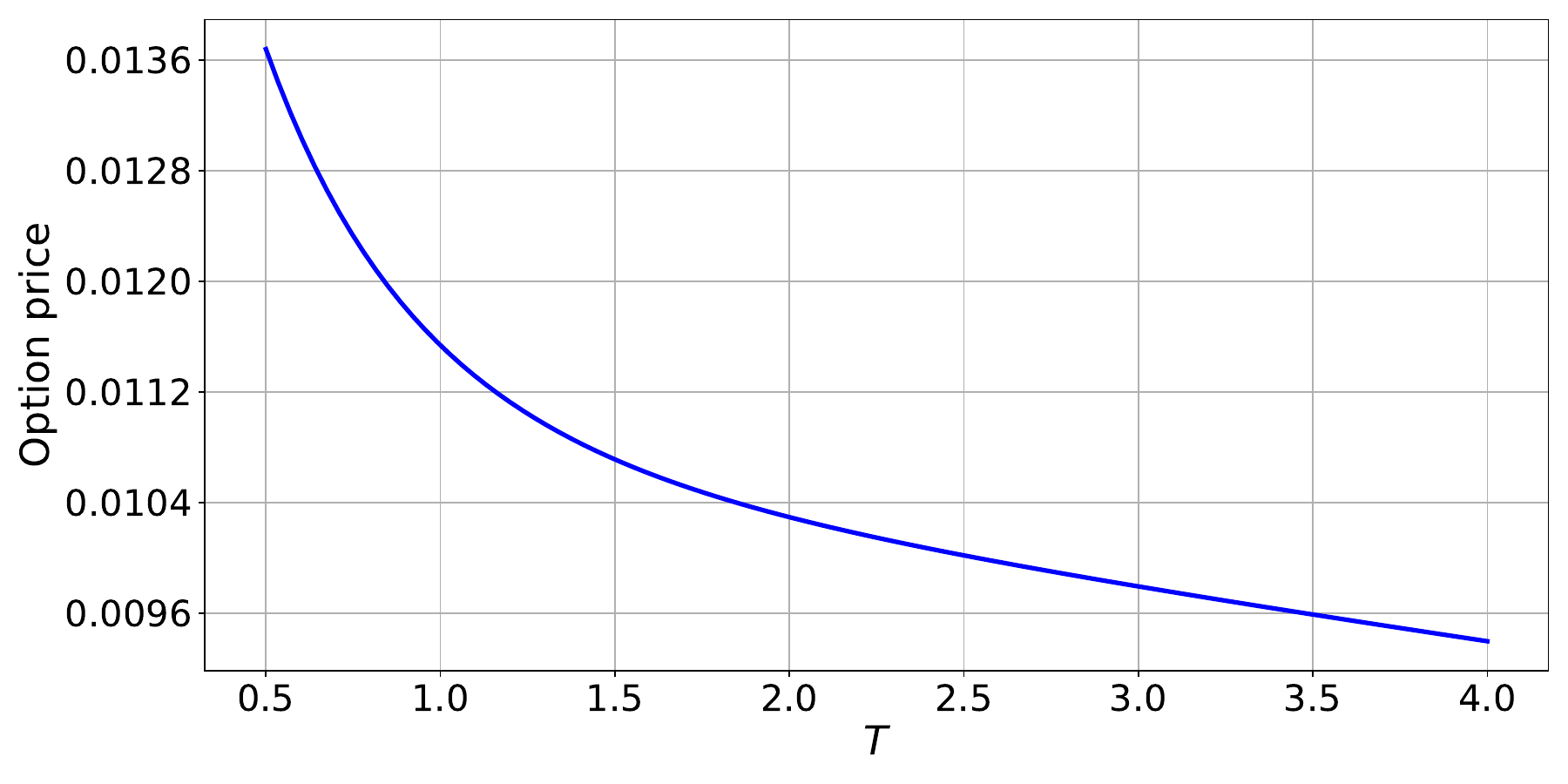}
	\end{center}

	{\footnotesize \textit{Note.} Each panel displays the sensitivity analysis of the guaranteed joint survival annuity option price with respect to different parameters. The top left panel displays $\alpha$. Top right panel displays $\beta$. Middle left panel displays $\rho$. Middle right panel displays $m_{11}$. Bottom left panel displays $\sigma_{11}$. Bottom right panel displays $T$, the exercise date of the guaranteed joint survival annuity option.}
\end{figure}

\begin{figure}[htbp]
	\caption{Comparison of approximation methods and errors.}\label{Fig:Approximation_methods}
	\begin{center}
		\includegraphics[width=7.5cm]{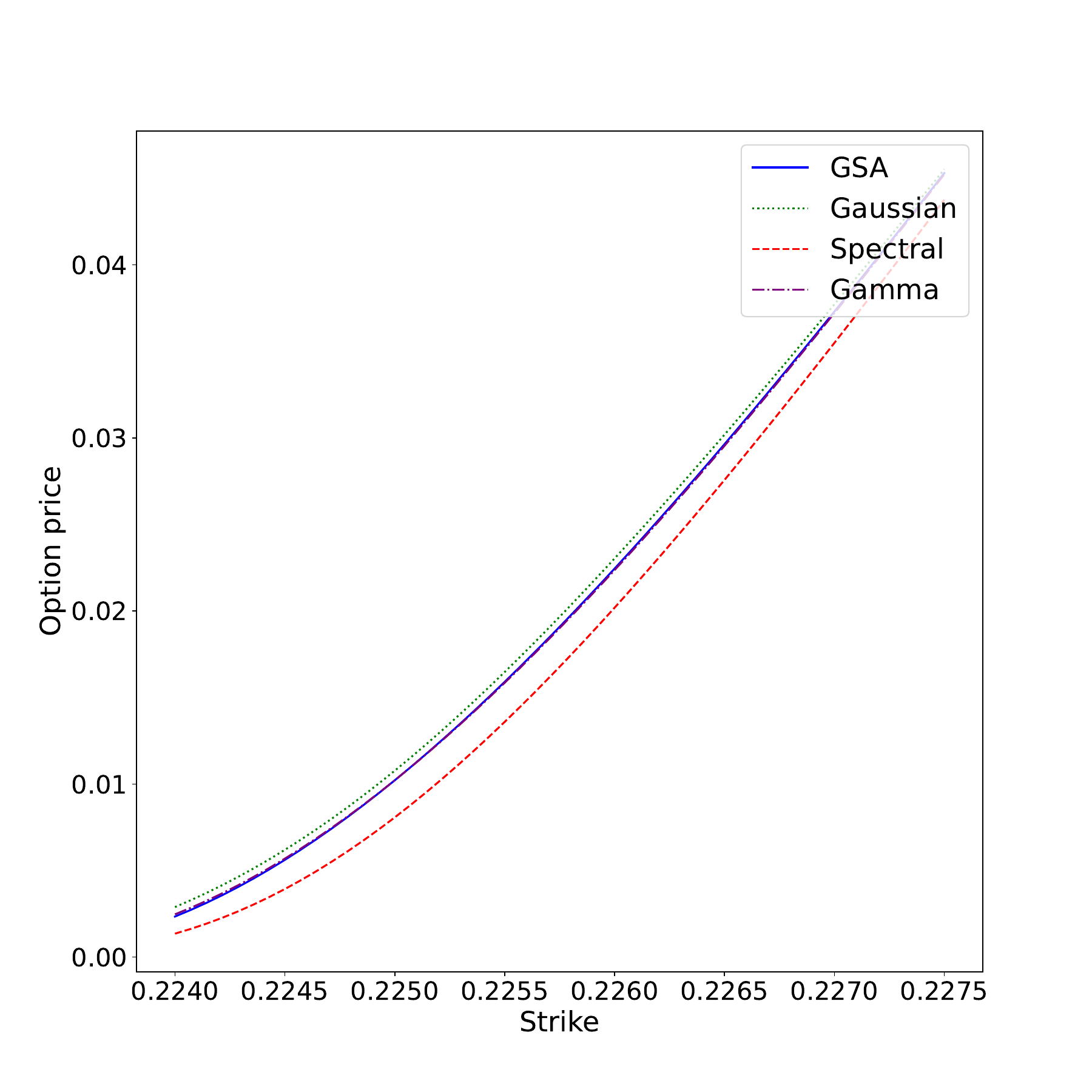}\includegraphics[width=7.5cm]{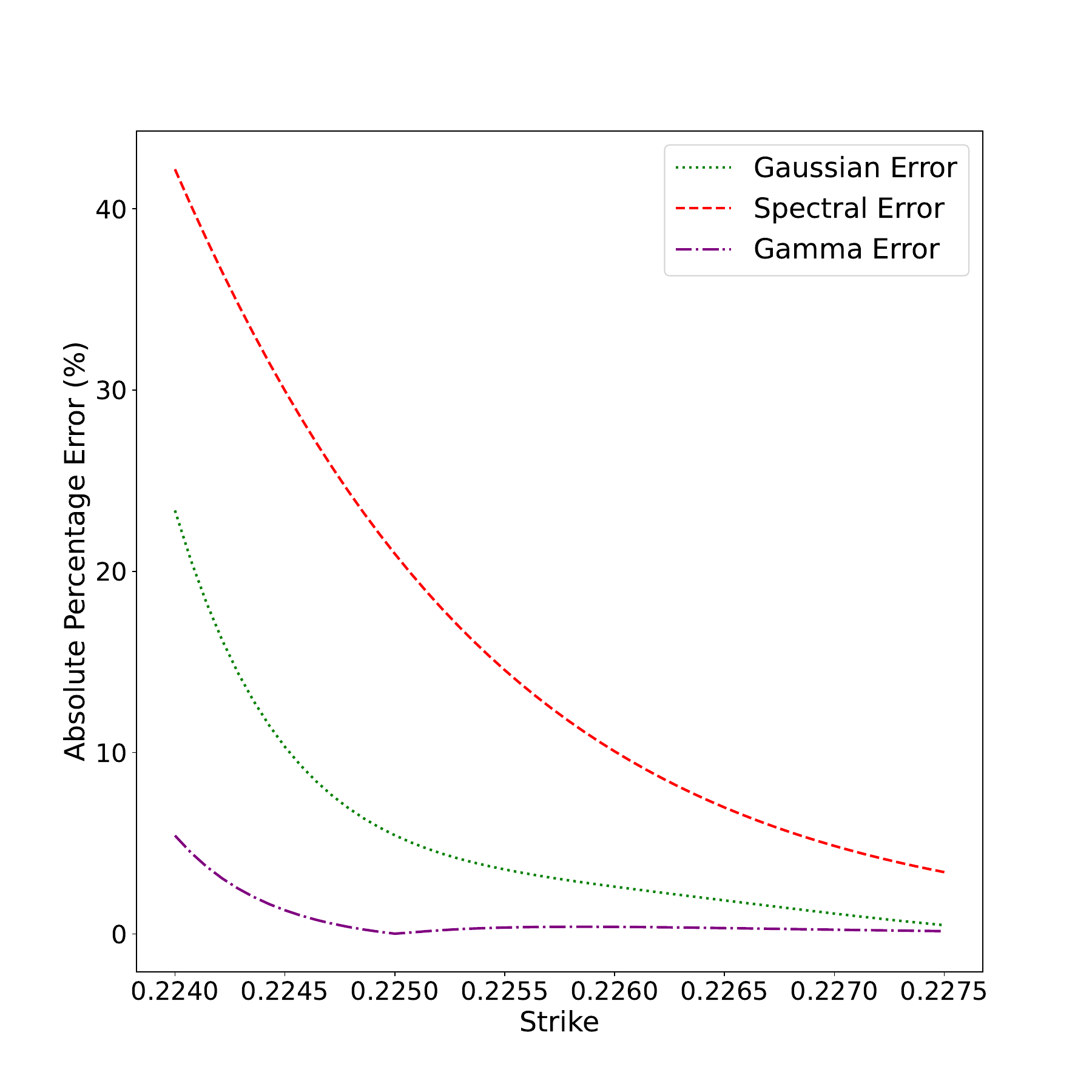}
	\end{center}
	{\footnotesize \textit{Note.} The left panel displays the guaranteed joint survival annuity option price across different strikes for different pricing methods: Characteristic function (blue solid), Gaussian approximation (green dotted), the spectral decomposition approximation (red dashed), and the gamma approximation (purple dotted-dashed). The right panel showcases the absolute percentage errors (\%) with respect to the characteristic function pricer across different strikes associated with the Gaussian approximation (green dotted), the spectral decomposition approximation (red dashed), and the gamma approximation (purple dotted-dashed).}
\end{figure}

\begin{figure}[htbp]
	\caption{Pricing difference for independence and dependence assumption.}\label{Fig:Independence_depdendence}
	\begin{center}
		\includegraphics[width=15cm]{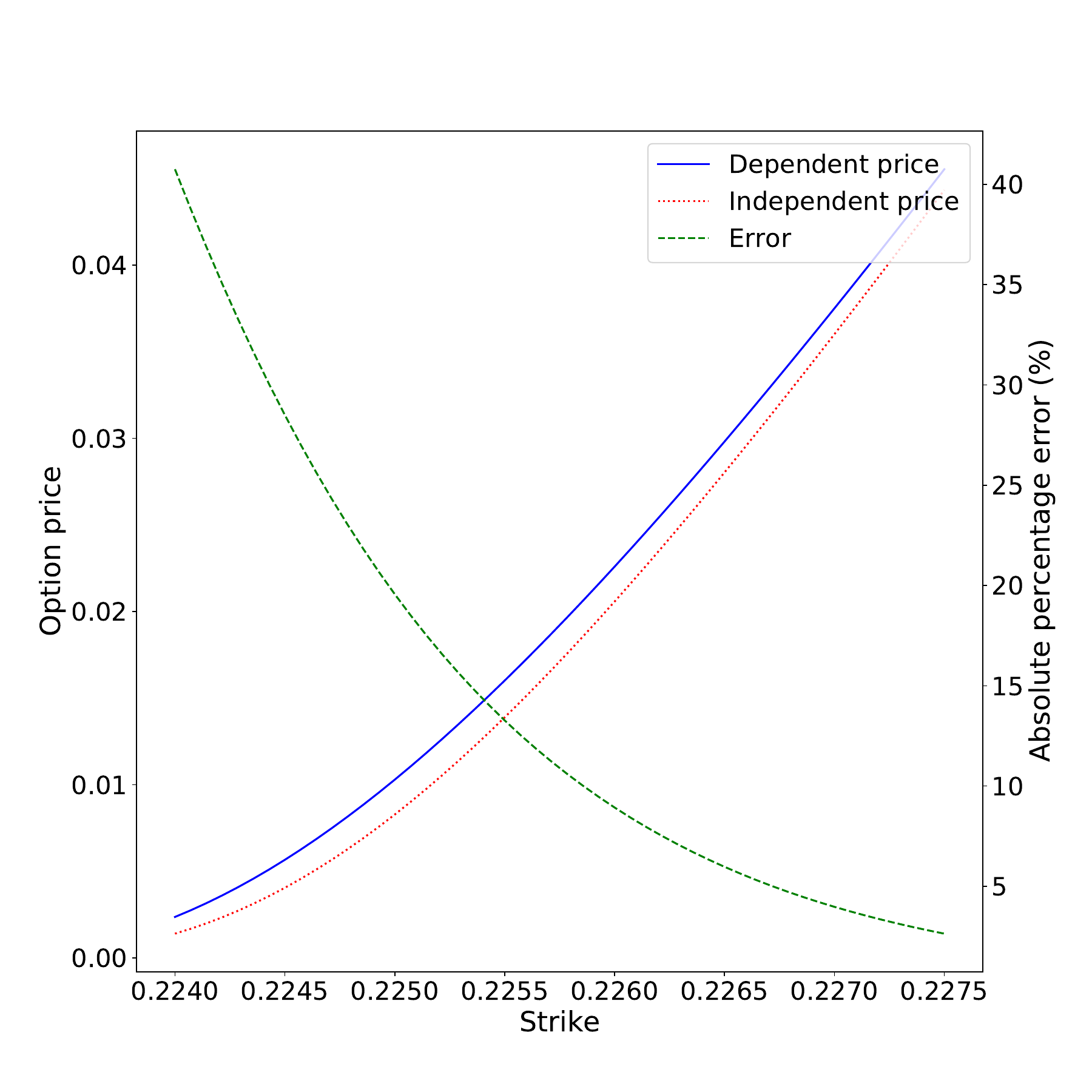}
	\end{center}
	{\footnotesize \textit{Note.} This panel displays the guaranteed joint survival annuity option price across different strikes under two different pricing assumptions: Dependence between annuitants given by the parameters outlined in Table~\ref{tab:model_param_values} and independence between annuitants which is also based on Table~\ref{tab:model_param_values} but with the matrix $\sigma$ defined by \eqref{EQ:sig_11_ind}-\eqref{EQ:sig_12_ind}.}
\end{figure}

%\footnotesize
%\linespread{1}
%\selectfont
\clearpage
\appendix
\section*{Supplementary Appendix}

%%%%%%%%%%%%%%%%%%%%%%%%%%%%%%%%%%%%%%%%%%%%%%%%%%%%%%%%%%%%%%%%%%%%%%%%%%%%%%%%%%%%%%%%%%%%%%%%%%%%%%%%%%%%%%%%%%%%%%%%%%%%%%%
\begin{proof}[Proof of Proposition \ref {Prop:SBZCPricing} ]
Using \eqref{EQ:SBZC} and the change of probability measure  \eqref{EQ:PricingKernel} and the expression for the state-price density \eqref{EQ:PricingKernel} we get 
\begin{align*}
\mathrm{SB}_0(t,T)	&=\mathbb{E}_t^\mathbb{Q}\left[e^{-\int_t^T(\mu_x(s)+\mu_y(s))ds} \right] \\
					&=\mathbb{E}_t\left[\frac{\zeta_T}{\zeta_t}\right] \\
					&=e^{-\alpha(T-t)}\frac{1+\mathbb{E}_t[\textup{tr}[u_0v_{T}]]}{1+\textup{tr}[u_0v_{t}]},
\end{align*}
and we conclude using Lemma \ref{LEM:TraceMean} with $u_0$.
\end{proof}
%%%%%%%%%%%%%%%%%%%%%%%%%%%%%%%%%%%%%%%%%%%%%%%%%%%%%%%%%%%%%%%%%%%%%%%%%%%%%%%%%%%%%%%%%%%%%%%%%%%%%%%%%%%%%%%%%%%%%%%%%%%%%%%
\begin{proof}[Proof of Proposition \ref{Prop:MortalityIntensityDistribution} ]
Define $\mathrm{x}_1=c_1 + \textup{tr}[h_1v_T]$ with $c_1:=\alpha/2 -\textup{tr}[u_1\omega]$, $h_1 = \alpha u_1-2u_1m$ and $\mathrm{x}_2=1 + \textup{tr}[u_0v_T]$ then the moment generating function of $(\mathrm{x}_1,\mathrm{x}_2)$ is 
\begin{align}
\Phi_{\mathrm{x}_1,\mathrm{x}_2}(z_1,z_2) 	&:=\mathbb{E}_t[e^{z_1\mathrm{x}_1 +z_2\mathrm{x}_2 } ]\nonumber\\
											&=e^{z_1c_1 + z_2}\Phi(T-t,z_1h_1+z_2u_0,v_t), \label{EQ:PhiX1X2}
\end{align}
with $\Phi(.,.,.)$ given in Proposition \ref{Prop:MGFBru}. As $\mu_x(T)$ is the ratio of these two variables, we can use \citet{Gurland1948} to obtain 
\begin{align}
\mathbb{P}(\mu_x(T)\leq z) = \frac{1}{2}-\frac{1}{\pi}\int_0^{+\infty}\Im \left(\frac{\Phi_{\mathrm{x}_1,\mathrm{x}_2}(\mathrm{i} s,-\mathrm{i} sz)}{s}\right) ds,
\end{align}
with $\mathrm{i}=\sqrt{-1}$. It gives the result after using \eqref{EQ:PhiX1X2}.
\end{proof}
%%%%%%%%%%%%%%%%%%%%%%%%%%%%%%%%%%%%%%%%%%%%%%%%%%%%%%%%%%%%%%%%%%%%%%%%%%%%%%%%%%%%%%%%%%%%%%%%%
\begin{proof}[Proof of Proposition \ref{Prop:JointSurvivalAnnuityDensity} ]
Define $\mathrm{x}_1 = b_3 + \tr\left[a_3 v_T\right]$ and $\mathrm{x}_2 = 1 + \tr\left[u_0 v_T\right]$. Then the moment generating function of $(\mathrm{x}_1, \mathrm{x}_2)$ is
\begin{align*}
	\Phi_{\mathrm{x}_1, \mathrm{x}_2}(z_1, z_2) &= \mathbb{E}_t\left[e^{z_1 \mathrm{x}_1 + z_2 \mathrm{x}_2} \right]\\
	&= e^{z_1 b_3 + z_2} \mathbb{E}_t\left[e^{\tr\left[\left(z_1 a_3 + z_2 u_0\right)v_T \right]} \right]\\
	&= e^{z_1 b_3 + z_2} \Phi(T-t, z_1 a_3 + z_2 u_0, v_t),
\end{align*}
with $\Phi(.,.,.)$ given in Proposition \ref{Prop:MGFBru}. Since $\sum_{i=1}^N\text{SB}(T,T_i)$ is the ratio of these two variables, we can use \citet{Gurland1948} to obtain 
\begin{align}
	\mathbb{P}\left(\sum_{i=1}^N\text{SB}(T,T_i)\leq z\right) = \frac{1}{2}-\frac{1}{\pi}\int_0^{+\infty}\Im \left(\frac{\Phi_{\mathrm{x}_1,\mathrm{x}_2}(\mathrm{i} s,-\mathrm{i} sz)}{s}\right) ds,
\end{align}
where, after using the above equation, we reach the result.
\end{proof}
%%%%%%%%%%%%%%%%%%%%%%%%%%%%%%%%%%%%%%%%%%%%%%%%%%%%%%%%%%%%%%%%%%%%%%%%%%%%%%%%%%%%%%%%%%%%%
\begin{proof}[Proof of Proposition \ref{GAOpricing}]
Starting from \eqref{EQ:GAO2}  and using \eqref{EQ:Radon} we have
\begin{align*}
\bar{C}(t,T,T_N) &= P(t,T) \mathbb{E}_t^{\mathbb{Q}}\left[ e^{-\int_{t}^{T}(\mu_x(s)+\mu_y(s))ds} \bar{C}(T,T,T_N) \right],\\   
				&= P(t,T)\mathbb{E}_t\left[\frac{\zeta_T}{\zeta_t} \left( \sum_{i=1}^N\mathrm{SB}(T,T_i)-1/g \right)_+  \right],
\end{align*}
and combining with \eqref{EQ:PricingKernel}, Proposition \ref{Prop:JointSurvivalAnnuity}  and some simplifications gives the result after defining $b_4(T,T_N)$ and $a_4(T,T_N)$ as in \eqref{EQ:B4} and \eqref{EQ:A4}.
\end{proof}
%%%%%%%%%%%%%%%%%%%%%%%%%%%%%%%%%%%%%%%%%%%%%%%%%%%%%%%%%%%%%%%%%%%%%%%%%%%%%%%%%%%%%%%%%%%%%%%%%%%%%
\begin{proof}[Proof of Proposition \ref{Prop:CumulantsExpansion}]
By definition the cumulants are $(\kappa_j)_{j \in \mathbb{N}}$ such that
\begin{align*}
\log \mathbb{E}_t[e^{zY_T}] = \sum_{j=1}^{+\infty} \kappa_j \frac{z^j}{j!}.
\end{align*}
From the expression \eqref{EQ:mgfBru} and using $\det(e^a) =\textup{etr}[a]$ it amounts to perform a series expansion of
\begin{align*}
z &\to \textup{tr}[\log(I_n -2z\varsigma_T a_4)]\\
z &\to \textup{tr}\left[\frac{\vartheta_T^\top}{2}(2z\varsigma_T a_4)(I_n-2z\varsigma_T a_4)^{-1}\right],
\end{align*}
and using $\log (I_n - x) = -\sum_{j=1}^{+\infty} \frac{x^j}{j}$ for $x\in \mathsf{M}(n)$ we get for the first function the series expansion
\begin{align*}
-\sum_{j=1}^{+\infty} (j-1)! 2^j\textup{tr}[(\varsigma_T a_4)^j] \frac{z^j}{j!},
\end{align*}
while $x(I_n-x)^{-1}=\sum_{j=1}^{+\infty}x^j$ leads for the other one to
\begin{align*}
\sum_{j=1}^{+\infty} j!2^{j-1}\textup{tr}[\vartheta_T^\top (\varsigma_T a_4)^j] \frac{z^j}{j!},
\end{align*}
and combining these two series expansions we get the result.
\end{proof}
%%%%%%%%%%%%%%%%%%%%%%%%%%%%%%%%%%%%%%%%%%%%%%%%%%%%%%%%%%%%%%%%%%%%%%%%%%%%%%%%%%%%%%%%%%%
\begin{proof}[Proof of Proposition \ref{GAOPricingApproximation}] 
If $Y_{T}=b_4(T,T_N) + \textup{tr}[a_4(T,T_N)v_T] $ with $ b_4(T,T_N)$ and $ a_4(T,T_N)$  as in ~\eqref{EQ:GAOpricingFormula} and suppose its density can be approximated by~\eqref{EQ:ApproximationDensity} then 
\begin{align*}
\mathbb{E}_t\left[   \left(  Y_{T}\right)_{+}\right] &\sim \sum_{j= 0}^3 \eta_j \int_0^{+\infty}\frac{1}{\sqrt{2\pi \kappa_2}} z(z-\kappa_1)^j e^{-\frac{(z-\kappa_1)^2}{2\kappa_2}}dz\,\\
&= \sum_{j = 0}^3 \eta_j \sqrt{\kappa_2}\kappa_2^{j/2} \int_{\frac{-\kappa_1}{\sqrt{\kappa_2}}}^{+\infty} z^{j+1} \frac{1}{\sqrt{2\pi }}e^{-\frac{z^2}{2}}dz + \sum_{j = 0}^3 \eta_j \kappa_1\kappa_2^{j/2} \int_{\frac{-\kappa_1}{\sqrt{\kappa_2}}}^{+\infty} z^j \frac{1}{\sqrt{2\pi }}e^{-\frac{z^2}{2}}dz\,\\
&= \sum_{j = 0}^3 \eta_j \xi_{j+1}+ \kappa_1\sum_{j = 0 }^{3} \eta_j \xi_{j} \,,
\end{align*}
where  $\eta_0 = 1$, $\eta_1 = -\frac{3}{\kappa_2^2}\frac{\kappa_3}{3!}$, $\eta_2 = 0$ and $\eta_3 = \frac{1}{\kappa_2^3}\frac{\kappa_3}{3!}$ while $\xi_j = \kappa_2^{j/2} \int_{\frac{-\kappa_1}{\sqrt{\kappa_2}}}^{+\infty}  \frac{z^j}{\sqrt{2\pi }}e^{-\frac{z^2}{2}}dz$ are given by $\xi_0 = N\left(\frac{\kappa_1}{\sqrt{\kappa_2}}\right)$, $\xi_1 =  \sqrt{\kappa_2} \varphi \left(\frac{\kappa_1}{\sqrt{\kappa_2}}\right)$, $\xi_2 = \kappa_2 N\left(\frac{\kappa_1}{\sqrt{\kappa_2}}\right) -\xi_1\kappa_1$, $\xi_3 = \xi_1(\kappa_1^2 +2\kappa_2)$ and $\xi_4 =  3\kappa_2^2 N\left(\frac{\kappa_1}{\sqrt{\kappa_2}}\right) - \xi_1(\kappa_1^3 +3\kappa_2\kappa_1)$ with $N(.)$ and $\varphi(.)$ the cumulative normal distribution and normal distribution density, respectively.
\end{proof}

%%%%%%%%%%%%%%%%%%%%%%%%%%%%%%%%%%%%%%%%%%%%%%%%%%%%%%%%%%%%%%%%%%%%%%%%%%%%%%%%%%%%%%%%%%%%%%%%%%%
\begin{proof}[Proof of Proposition \ref{Prop:MGFChiSquared}]
Let $\gamma$ as stated above and consider $\gamma^\top v_T \gamma=\textup{tr}[\gamma\gamma^\top v_T]$ then for $z\in \mathbb{R}$ using \eqref{EQ:mgfBru} we get
\begin{align*}
\mathbb{E}_t[e^{z\textup{tr}[\gamma\gamma^\top v_T]}]=\frac{\textup{etr}\left(\frac{\vartheta_T^\top}{2}(2z\varsigma_T \gamma\gamma^\top)(I_n-2z\varsigma_T\gamma\gamma^\top)^{-1}  \right)}{\det(I_n -2z\varsigma_T\gamma\gamma^\top)^{\beta/2}}\,,
 \end{align*}
and since $\gamma\gamma^\top$ is rank one there exists a basis such that $\gamma\gamma^\top$  in that basis is $e_{11}$ so that $\det(I_n -2z\varsigma_Te_{11})^{\beta/2}=(1-2z(\varsigma_T)_{11})^{\beta/2}$ ($\varsigma_T$ has to be expressed in that basis) and that the numerator above is
\begin{align*}
\textup{etr}\left(\frac{\vartheta_T^\top}{2}(2z\varsigma_T e_{11})(I_n-2z\varsigma_Te_{11})^{-1}  \right)=\exp\left(\frac{(\vartheta_T^\top)_{11}  z(\varsigma_T)_{11}}{1-2z(\varsigma_T)_{11}}\right)\,,
\end{align*}
with both $\varsigma_T$ and $\vartheta_T$ expressed in that basis. Combining the two it leads to
\begin{align*}
\mathbb{E}_t[e^{z\gamma^\top v_T \gamma}]=\frac{e^{\frac{(\vartheta_T^\top)_{11} z(\varsigma_T)_{11}}{1-2z(\varsigma_T)_{11}}}}{(1-2z(\varsigma_T)_{11})^{\beta/2}}\,,
 \end{align*}
from which we get the result.
\end{proof}
%%%%%%%%%%%%%%%%%%%%%%%%%%%%%%%%%%%%%%%%%%%%%%%%%%%%%%%%%%%%%%%%%%%%%%%%%%%%%%%%%%%%%%%%%%%%%%%%%%%%%%%

\begin{proof}[Proof of Proposition \ref{Prop:ProjectionOnEig}]
Let $Y_T = b_4(T,T_N)+\textup{tr}[a_4(T,T_N)v_T]$ and rewrite $\textup{tr}[a_4(T,T_N)v_T]$ as $\textup{tr}[(\frac{a_4+a_4^\top }{2})v_T]$, and since  $\frac{a_4+a_4^\top }{2}\in \mathbb{S}_n$, it admits the decomposition \eqref{EQ:SpectralA4}. But $\textup{tr}[\gamma_i\gamma_i^\top v_T] = \gamma_i^\top v_T \gamma_i$ and according to Proposition \ref{Prop:MGFChiSquared}, $\gamma_i^\top v_T \gamma_i/(\varsigma_T)_{ii}$ has a non-central Chi-squared distribution with degrees of freedom $\beta$ and non-centrality parameter $(\vartheta_T^\top)_{ii}$. Note that these variables are not independent, but we can approximate $Y_T$, although we do not have a control of the approximation error, with \(\tilde{Y}_T\) defined by \eqref{eq:tildeY_T}. This set of variables follows a generalized Chi-squared distribution (since it is a collection of independent scalar non-central Chi-squared variables).
\end{proof}

%%%%%%%%%%%%%%%%%%%%%%%%%%%%%%%%%%%%%%%%%%%%%%%%%%%%%%%%%%%%%%%%%%%%%%%%%%%%%%%%%%%%%%%
\begin{proof}[Proof of Proposition \ref{GAOPricingApproximationGamma} ] 
We have $Y_{T}=b_4  + \textup{tr}[a_4v_T] $ and if the eigenvalues of $a_4$ are all positive then $\textup{tr}[a_4v]>0$ for $v\in \mathbb{S}_n^{++}$ so that necessarily $b_4<0$ (otherwise $Y_{T}$ would be positive and the option would end always in the money) so that there exists $k<0$ such that $Y_{T}>k$. A natural choice for $k$ is $k=b_4$. Then if we denote $Z= Y_{T}-k$ we can rewrite  as 
\begin{align}
\mathbb{E}_t\left[   \left(  Y_{T}\right)_{+}\right] = \mathbb{E}_t\left[   \left(  Z + k \right)_{+}\right]\,,
\end{align}
with $\mu_i^Z=\mathbb{E}_t\left[   Z^i\right] =\sum_{j=0}^i \binom{i}{j}(-1)^j\mu_{i-j}k^j$ and $\mu_{i}= \mathbb{E}_t\left[   Y_{T}^i\right]$ the $i'{th}$ moment of $Y_{T}$ that can be deduced from the cumulants of $Y_{T}$ thanks to \eqref{eq:MomentsFromCumulants}. Note that as $Z>0$ by construction then $\mu_1^Z>0$ and we approximate its distribution by a perturbation of the (scalar) gamma distribution following \citet{FilipovicMayerhoferSchneider2013}. The density of $\frac{Z\mu_1^Z}{\mu_2^Z-(\mu_1^Z)^2}$ can be approximated at order three by 
\begin{align}
g^{3}(z) = w(z) \sum_{j=0}^{3} c_{j} H_{j}(z)\,,
\end{align}
with $w$ the auxiliary density function $w(z) = \frac{z^{\bar{\alpha}} e^{-z}}{\Gamma(1+\bar{\alpha})}$, that is the (scalar) gamma distribution with shape $1+\bar{\alpha}$ and rate 1, some known constants $c_{j}$, and  $\left\{H_{j} \left.\right| j =0,\ldots, 3 \right\}$ an orthonormal basis of polynomials. \citet{FilipovicMayerhoferSchneider2013} prove that the functions $H_{j}$ are given by $H_{j}=\frac{{\tilde{H}_{j}} }{{\left\|\tilde{H}_{j}(z) \right\|}}$, with ${\left\|\tilde{H}_{j}(z) \right\|} =\sqrt{\frac{\prod_{i=1}^{j}(i+\bar{\alpha})}{j!}} $ and
\begin{align}
\tilde{H}_{0}(z) &= 1 \,, \\
\tilde{H}_{1}(z) &= -z +\bar{\alpha} + 1 \,, \\
\tilde{H}_{2}(z) &=\frac{1}{2}( z^{2} -2z(\bar{\alpha}+2) +\bar{\alpha}^2 +3\bar{\alpha} +2) \,, \\
\tilde{H}_{3}(z) &=\frac{1}{6}( -z^{3} +3z^{2}(\bar{\alpha}+3)  -3z(\bar{\alpha}^2 +5\bar{\alpha} +6) +\bar{\alpha}^3 +6\bar{\alpha}^2 +11\bar{\alpha} +6) \,.
\end{align}
Following \citet[Section 7.1]{FilipovicMayerhoferSchneider2013}, we choose $\bar{\alpha}= \frac{(\mu_{1}^{Z})^{2}}{\mu_{2}^Z- (\mu_{1}^Z)^{2}} -1 $ and $\bar{\beta}= \frac{\mu_{1}^Z}{\mu_{2}^Z-(\mu_{1}^Z)^{2}}>0$ as $\mu_1^Z>0$ then the coefficients $\lbrace c_{i};i=0,\ldots, 3 \rbrace$ are given by: $c_{0} = 1$, $c_{1} = 0$, $c_{2} =0$ and
\begin{align}
c_{3} &=\frac{(\bar{\alpha}+1)\left( (\bar{\alpha}+2)(\bar{\alpha}+3)-\frac{(\bar{\alpha}+1)^2 \mu_{3}^Z }{(\mu_{1}^Z)^{3}}\right) }{\sqrt{6}\sqrt{(\bar{\alpha}+1)(\bar{\alpha}+2)(\bar{\alpha}+3)}} \,.
\end{align}
Eventually, the density of $Z$, taking into account the change of variable $\frac{Z\mu_1^Z}{\mu_2^Z-(\mu_1^Z)^2}$, can be approximated by 
\begin{align}
g^{3}(z) = w(\bar{\beta}z) \sum_{j=0}^{3} c_{j} H_{j}(\bar{\beta} z)\bar{\beta}\,,
\end{align}
and therefore, by rewriting $H_{j}(z)=\sum_{i=0}^{j} \gamma_{j,i} z^i$, we get 
\begin{align}
\mathbb{E}_t\left[   \left(  Y_{T}\right)_{+}\right]&=\mathbb{E}_t\left[   \left(  Z+k\right)_{+}\right] \sim \int_{-k}^{+\infty} (z+k)g^{3}(z)dz\\
&= \int_{-\bar{\beta} k}^{+\infty}  \left(\frac{y}{\bar{\beta}} + k\right)  w(y)\sum_{j=0}^{3}\sum_{i=0}^{j}  c_j\gamma_{j,i}y^{i}dy\\
&=  \sum_{j=0}^{3}\sum_{i=0}^{j}  c_j\gamma_{j,i}\frac{\Gamma(\bar{\alpha}+i+2,-\bar{\beta} k)}{\bar{\beta}\Gamma(1+\bar{\alpha})}  + \sum_{j=0}^{3}\sum_{i=0}^{j}  c_j\gamma_{j,i}k\frac{\Gamma(\bar{\alpha}+i+1,-\bar{\beta} k)}{\Gamma(1+\bar{\alpha})}\,,
\end{align}
where we used the condition $\bar{\beta}>0$ while $\Gamma(s,x)$ stands for the upper incomplete gamma function. To avoid overflows with the gamma function, we use its regularized version so that
\begin{align}
\frac{\Gamma(\bar{\alpha}+i+2,-\bar{\beta} k)}{\Gamma(1+\bar{\alpha})}=Q(\bar{\alpha}+i+2,-\bar{\beta} k)\frac{\Gamma(\bar{\alpha}+i+2)}{\Gamma(1+\bar{\alpha})}\,,
\end{align} 
and thank to the relations $Q(s,x)=\Gamma(s,x)/\Gamma(s)$ and $(x)_n=\Gamma(x+n)/\Gamma(x)$ where $(x)_n$ is the Pochhammer function (rising factorial) we deduce the result.\\

If the eigenvalues of $a_4$ are all negative then $\textup{tr}[a_4v]<0$ for $v\in \mathbb{S}_n^{++}$ so that necessarily $b_4>0$ and there exists $k>0$ such that $Y_{T}<k$. A natural choice for $k$ is $k=b_4$. To reuse previous computations, we rewrite the expectation as follows
\begin{align}
\mathbb{E}_t\left[   \left(  Y_{T}\right)_{+}\right] = \mathbb{E}_t\left[   \left( - (-b_4 - \textup{tr}[a_4v_T]) \right)_{+}\right]
\end{align}
and define $Z= (-b_4 - \textup{tr}[a_4v_T])+k$, then $\mathbb{E}_t\left[   \left(  Y_{T}\right)_{+}\right] = \mathbb{E}_t\left[   \left(-(Z+k)\right)_{+}\right]$ and using the previous case we reach the announced result.
\end{proof}

\end{document}